\documentclass[11pt]{article}

\usepackage[margin=1.1in]{geometry}

\usepackage{amsmath,amssymb,amsthm,mathtools}
\usepackage{bm}
\usepackage{booktabs}
\usepackage{enumitem}
\usepackage{algorithm}
\usepackage{algpseudocode}
\usepackage[expansion=false]{microtype}
\usepackage[colorlinks=true,linkcolor=blue!50!black,citecolor=blue!50!black,urlcolor=blue!50!black]{hyperref}
\usepackage{xcolor}
\usepackage{tikz}
\usetikzlibrary{positioning,calc,fit,arrows}
\tikzset{
  site/.style={draw, rounded corners=1pt, minimum height=5.2mm,
               minimum width=15mm, font=\scriptsize, inner sep=2pt},
  ev/.style={draw, minimum height=5.2mm, minimum width=15mm,
             font=\scriptsize, inner sep=2pt, fill=black!4},
  stage/.style={draw, rounded corners=2pt, align=center,
                font=\scriptsize, inner sep=3.5pt, minimum height=7mm},
  lab/.style={font=\scriptsize\itshape},
  tinylab/.style={font=\tiny},
  flow/.style={->, >=stealth, semithick},
  okline/.style={->, >=stealth, thick, black!55},
  badline/.style={->, >=stealth, thick, densely dashed},
}

\newtheorem{theorem}{Theorem}[section]
\newtheorem{lemma}[theorem]{Lemma}
\newtheorem{proposition}[theorem]{Proposition}
\newtheorem{corollary}[theorem]{Corollary}
\theoremstyle{definition}
\newtheorem{definition}[theorem]{Definition}
\newtheorem{assumption}{Assumption}
\newtheorem{remark}[theorem]{Remark}

\newcommand{\WUR}{\mathrm{WUR}}
\newcommand{\Rb}{R_{\mathrm{b}}}
\newcommand{\Rf}{R_{\mathrm{f}}}
\newcommand{\pref}{\theta}
\newcommand{\front}{\varphi}
\newcommand{\match}{\mu}
\newcommand{\sev}{\mathrm{sev}}
\newcommand{\prot}{\mathrm{P}}
\newcommand{\intf}{\mathrm{I}}
\newcommand{\migr}{\mathrm{M}}
\newcommand{\ind}[1]{\mathbb{1}\!\left[#1\right]}
\newcommand{\E}{\mathbb{E}}
\newcommand{\Prob}{\mathbb{P}}

\title{\textbf{A Telemetry-Driven Model for Quantifying\\ Upgrade Risk in Durable Workflow Execution}\\[0.6em]
\large Version-skew risk inference for event-sourced workflows\\ without dry-run execution}

\author{%
Luca Maraschi \qquad Matteo Collina\\[0.2em]
\normalsize Platformatic Inc.\\
\normalsize \texttt{\{luca,matteo\}@platformatic.dev}}

\date{July 2026}

\begin{document}
\maketitle

\begin{abstract}
Durable workflow engines reconstruct execution state by deterministically
replaying an immutable event log, coupling every in-flight run to the code
version that produced its history: a new deployment can invalidate the replay
of runs started under the old version, silently corrupting state or halting
progress. Existing mitigations---pinning, patch gates, side-by-side
deployment---treat every change as maximally dangerous and drain old versions,
untenable for workflows that sleep for weeks.

We present a closed-form probabilistic model that quantifies the risk of
upgrading in-flight runs from workflow version $V_1$ to $V_2$ using only a
static structural diff and telemetry the protocol already persists---event
logs, step payloads, historical paths---with no dry-run, sandbox, or shadow
execution. Risk decomposes along three axes (protocol, interface, state
migration) and combines an \emph{exact} backward (rehydration) term, computed
on recorded prefixes modulo trace equivalence of concurrent completions, with
a \emph{probabilistic} forward term from hitting probabilities in an
empirically estimated Markov model of control flow. Estimation is Bayesian
throughout, so the Workflow Upgrade Risk ($\WUR$) score carries a credible
interval and thin telemetry surfaces as uncertainty. We prove that a zero
backward-risk verdict certifies safe rehydration under the new version, and
derive a policy partitioning runs into \emph{migrate}, \emph{review}, and
\emph{pin} classes. Finally we drop the inter-run independence assumption:
coupling through hooks, hierarchy, and shared resources is captured by an
empirical coupling graph, fleet risk becomes the least fixpoint of a
failure-contagion operator, and the coupling-aware migrate/pin partition is
computed exactly as a minimum $s$--$t$ cut.
\end{abstract}

\section{Introduction}

Durable execution frameworks---Temporal~\cite{temporal}, Azure Durable
Functions~\cite{durablefunctions}, and more recently the Workflow
SDK~\cite{wdk}---let developers express long-running, stateful
processes as ordinary code. The engine persists an append-only log of events
(step invocations, results, timers, external signals); when a process must
resume after a crash, a suspension, or a redeployment, the engine re-executes
the orchestration function from the beginning and substitutes recorded results
for completed steps. Correctness of this \emph{replay} rests on a determinism
contract: given the same event history, the orchestration code must make the
same decisions.

The contract has a corollary that is under-examined in both literature and
practice: the event log is only meaningful \emph{relative to the code that
produced it}. When version $V_2$ of a workflow is deployed while runs started
under $V_1$ are still in flight, three distinct failure modes arise:

\begin{enumerate}[label=(F\arabic*),leftmargin=3em]
  \item \textbf{Replay divergence.} $V_2$ reaches a different step than the
  log expects (a step was inserted, removed, or reordered). Depending on the
  engine this manifests as a hard nondeterminism error or---worse---as a
  silently mismatched result substitution.
  \item \textbf{Rehydration failure.} The recorded payload of a completed
  step no longer satisfies the shape that $V_2$'s downstream code consumes
  (a field was removed, renamed, or its type narrowed).
  \item \textbf{Behavioral drift.} Replay succeeds, but steps not yet
  executed behave differently under $V_2$; the run completes with semantics
  neither version's author intended.
\end{enumerate}

The state of the art is operational, not analytical. Temporal offers
worker versioning and an explicit \texttt{patched()} API~\cite{temporalver}
that gates code paths on a run's provenance; platform-level skew protection
pins runs to the deployment that created them and drains old versions.
These mechanisms are sound but uniformly pessimistic: they assume every
change is incompatible with every in-flight run. For short-lived runs the
cost is idle capacity; for workflows that legitimately sleep for
months---subscription lifecycles, compliance timers, drip
campaigns---pessimism means a critical bug fix cannot reach a stranded run
at all.

This paper develops the missing analytical layer. We ask: \emph{given the
diff between $V_1$ and $V_2$ and the telemetry the protocol already
persists, what is the probability that upgrading a given in-flight run
fails, and how does that risk aggregate over the fleet?} Our central
observation is that event-sourced workflow engines are an unusually
favorable setting for change-risk inference. In a general distributed
system, risk models must infer interaction patterns from sampled traces;
here, the protocol persists \emph{complete} histories---every path ever
taken, every payload ever produced---as a byproduct of its durability
guarantee. Consequently the backward-looking component of upgrade risk can
be computed \emph{exactly}, and only the forward-looking component requires
probabilistic estimation.

\paragraph{Contributions.}
\begin{enumerate}[leftmargin=2em]
  \item A formal model of the durable-workflow protocol (event log, replay
  relation, determinism contract) sufficient to state upgrade compatibility
  precisely (\S\ref{sec:model}).
  \item A three-axis change taxonomy---protocol, interface,
  migration---with per-node incompatibility functions computable from the
  static diff and payload telemetry (\S\ref{sec:changes}).
  \item The $\WUR$ risk model: an exact backward term over recorded
  prefixes and a forward term based on hitting probabilities in an
  empirically estimated absorbing Markov chain, composed by a noisy-or
  operator (\S\ref{sec:risk}).
  \item A Bayesian estimation procedure with credible intervals, using
  Dirichlet posteriors on transitions and Beta posteriors on payload
  compatibility, pooled across axes by confidence-weighted log-odds
  (\S\ref{sec:estimation}).
  \item Soundness and monotonicity results: $\Rb(r)=0$ under conservative
  matching certifies safe rehydration of run $r$; a
  \emph{differential-demand lemma} confining contract inference to changed
  consumer slices; and a canonical-form procedure deciding prefix validity
  modulo trace equivalence of concurrent completions in
  $O(L \log L)$ (\S\ref{sec:properties}).
  \item An $O(|\Delta|\cdot|\pref|_{\max}\cdot|\mathcal{R}| + |N|^3)$
  algorithm and a threshold policy that partitions the fleet into
  migrate/review/pin classes with expected-loss ranking
  (\S\ref{sec:algorithm}--\ref{sec:policy}).
  \item An inter-run coupling extension (\S\ref{sec:coupling}): fleet
  risk as the least fixpoint of a failure-contagion operator on an
  empirically constructed coupling graph; mixed-version window risk
  scored by the interface machinery applied \emph{across} channels and
  weighted by logged event rates; and an exact coupling-aware
  migrate/pin partition via $s$--$t$ minimum cut.
\end{enumerate}

These contributions form three deliberate layers, each consuming the
previous one's output and nothing else: a \emph{certification layer}
(\S\ref{sec:model}--\ref{sec:properties}) that decides, per run, what
recorded history already proves; a \emph{risk layer}
(\S\ref{sec:risk}--\ref{sec:policy}) that quantifies what history can
only make probable; and an \emph{operations layer}
(\S\ref{sec:coupling}) that turns per-run verdicts into a fleet
decision under coupling. A reader interested only in the certification
result can stop after \S\ref{sec:properties}; an operator deploying the
model needs all three, because per-run verdicts that ignore coupling
recommend partitions that violate saga atomicity in the mixed-version
window.

Everything in the model is world-agnostic: it consumes only the protocol
surface (run listings, event logs, recorded payloads) and never touches a
specific persistence backend or executes workflow code.

\section{Related Work}\label{sec:related}

\paragraph{Durable execution and versioning.}
Temporal's versioning facilities~\cite{temporalver} (patch markers, worker
build IDs, deployment-based routing) and platform-level skew protection are
\emph{mechanisms} for avoiding incompatible replay; they provide no measure
of whether a given change is in fact incompatible, and default to pinning.
Azure Durable Functions documents versioning as a set of manual
strategies~\cite{durablefunctions}. Our model is complementary: it supplies
the decision function that these mechanisms lack.

\paragraph{Event sourcing and schema evolution.}
Event-sourced systems~\cite{fowler-es} face upgrade problems for
\emph{data}: old events must remain readable by new code. Schema-evolution
disciplines in serialization frameworks (Avro, Protocol Buffers) define
compatibility lattices over record types~\cite{kleppmann}. We adopt the
lattice view for the interface axis but extend the problem to
\emph{control flow}: in durable workflows the event log encodes not only
data but the shape of a partially executed computation.

\paragraph{Change-risk prediction.}
Defect- and risk-prediction models in software engineering estimate failure
probability from change metrics (churn, coupling,
history)~\cite{nagappan}. These models are statistical over code artifacts
and outcomes; they do not exploit an execution protocol. Our setting
inverts the balance: the protocol's determinism contract makes a large part
of the risk \emph{computable} rather than merely predictable.

\paragraph{Markov models of program behavior.}
Using empirically estimated Markov chains to model control flow and compute
reachability or hitting probabilities is classical~\cite{norris,whittaker}.
We apply the machinery to the step graph of a workflow, with the notable
property that the training data (complete path histories) is exhaustive
rather than sampled.

\paragraph{Energy minimization on graphs.}
Exact MAP inference for binary pairwise Markov random fields via minimum
cut is classical~\cite{greig,kz}; multi-label extensions rely on
move-making approximations~\cite{bvz}. We import this machinery for the
coupling-aware migrate/pin partition, where submodularity of the pairwise
terms holds by construction.

\subsection{The Engine Landscape}\label{sec:engines}

The model is generic over durable execution engines, but genericity is
earned per engine, not assumed: what varies is the \emph{site identity
discipline} that the matching $\match$ of Definition~\ref{def:matching}
must reproduce, the strength of the determinism contract
(Assumption~\ref{ass:det}), and the telemetry surface
(Assumption~\ref{ass:log}). Table~\ref{tab:engines} fixes these
parameters for the engines we are aware of; the rest of the paper refers
to this table rather than re-litigating engine differences at each use.

\begin{table}[t]
\centering\small
\setlength{\tabcolsep}{4pt}
\begin{tabular}{p{2.1cm}p{2.9cm}p{2.5cm}p{2.7cm}p{2.6cm}}
\hline
engine & site identity & determinism & native versioning & payload
telemetry \\
\hline
Workflow SDK~\cite{wdk} & global invocation ordinal $+$
module-qualified name (one ID sequence for steps, timers, hooks;
verified, Rem.~\ref{rem:conformance}) & virtualized VM: seeded RNG, log-driven
clock & runs pinned to starting deployment; explicit cancel{+}rerun,
rerun-on-latest, continue-as-new & full event log with
payloads (\texttt{runs.list}, \texttt{events.list}) \\[2pt]
Temporal~\cite{temporal} & command sequence matched positionally against
history, type-checked & SDK-enforced APIs; sandbox in TS & patch markers,
worker build IDs, pinned vs.\ auto-upgrade
routing~\cite{temporalver} & full histories via visibility/history
APIs \\[2pt]
Azure Durable Functions~\cite{durablefunctions} & task scheduling order
against history, name-checked & replay-safe API discipline & side-by-side
deployment guidance (pin by default) & history tables \\[2pt]
Restate~\cite{restatever} & journal entries matched positionally
(mismatch $\Rightarrow$ RT0016) & \texttt{ctx.run} journaling of
nondeterminism & immutable deployments; invocations pinned by default,
manual move & journal introspection (SQL) \\[2pt]
Cloudflare Workflows~\cite{cfworkflows} & step \emph{name} as cache key
$+$ occurrence count & developer discipline; orchestrator code outside
steps re-executes & engine-managed; instances resume on current code &
per-step persisted state \\[2pt]
AWS Step Functions~\cite{sfn} & state-machine definition (no code
replay) & n/a (declarative) & definition versions and aliases; executions
finish on their version & execution history \\
\hline
\end{tabular}
\caption{Engine parameters for the model. ``Site identity'' is the
discipline $\match$ must preserve; ``native versioning'' is the mechanism
the engine offers, for which this paper supplies the missing decision
function.}
\label{tab:engines}
\end{table}

Three observations structure the comparison. First, the
\emph{ordinal-identity family} (Workflow SDK, Temporal, Azure Durable
Functions, Restate) shares the strict discipline formalized in
\S\ref{sec:model}: any insertion, deletion, or reordering of sites ahead
of recorded events is fatal, replay failures are hard errors, and the
model applies with $\match$ as defined. It also shares pin-by-default
versioning: Workflow SDK pins every run to the deployment that started
it, and its versioning guide's recovery procedure for buggy
deployments---find the affected runs, cancel, rerun on the latest
deployment---presupposes exactly the affected-run determination this
model computes, while its \texttt{deploymentId:\ "latest"} escape
hatch explicitly delegates the cross-deployment compatibility judgment
to the developer. Restate's documentation is the
most explicit ancestor of our taxonomy, enumerating as unsafe exactly the
reordering, insertion/removal, input-change, and conditional-logic
changes that populate $\Delta_\prot$ and $\Delta_\intf$---but, like
Temporal's pinned-versus-auto-upgrade choice and Durable Functions'
side-by-side guidance, it leaves the per-run decision to human judgment.
These engines pay for safety with pin-by-default: old deployments must be
kept alive until their pinned runs drain, with the forgone-fix exposure
$\beta_r$ of \S\ref{sec:policy} accruing silently. Temporal's patch
API deserves its own sentence, since it is the most sophisticated native
mechanism: \texttt{patched()} lets one worker serve both versions by
branching in workflow code, which solves mixed-version execution but
creates a new liability---patches accumulate, and knowing when a patch
is removable requires knowing that no in-flight run can still traverse
its old branch. That is a hitting probability on recorded frontiers,
i.e., precisely the forward machinery of \S\ref{sec:risk}; the model
answers ``which runs force this patch to exist'' with a number. Our
model is, in general, the decision procedure these engines' escape
hatches presuppose.

Second, \emph{name-keyed identity} (Cloudflare Workflows) trades
strictness for tolerance: steps are memoized by name with an occurrence
counter, so inserting or reordering \emph{differently named} steps is
tolerated where the ordinal family diverges---but a renamed step is a
cache miss that silently \emph{re-executes}, converting failure class F1
(hard divergence) into F2-like duplicated side effects whose severity is
governed by idempotency, precisely the quantity \S\ref{sec:severity}
estimates from retry telemetry. The framework accommodates this by
parameterizing $\match$: name-multiset matching with occurrence indices
replaces ordinal preservation, $\Delta_\prot$ coarsens, and the severity
axis does proportionally more work. A further difference widens
$\Delta_\prot$ in the other direction: orchestrator code outside steps
is re-executed rather than virtualized, so \emph{any} change to it is
protocol-relevant, not only changes at sites.

Third, \emph{definition-based execution} (Step Functions) versions the
state machine rather than replaying code: executions complete on the
definition they started with, so backward risk trivializes (there is no
rehydration of recorded state into new code) and the model reduces to its
forward-risk and coupling components applied at the boundary between
definition versions.

The remainder of the paper develops the model for the ordinal-identity
discipline and verifies it against Workflow SDK as the reference
instantiation; the parameterization above delimits what transfers
elsewhere and at what cost.

\section{System Model}\label{sec:model}

We formalize the protocol implemented by Workflow SDK-class engines. The model is
deliberately minimal; any engine satisfying
Assumptions~\ref{ass:det}--\ref{ass:log} is covered.

\begin{definition}[Workflow version]\label{def:version}
A \emph{workflow version} $V$ is a deterministic orchestration program
whose observable actions are drawn from a finite set $N_V$ of \emph{sites}:
step invocations, timer starts (\texttt{sleep}), and external-event waits
(hooks). Each site $v \in N_V$ carries a label
$\ell(v) = (\mathrm{name}(v), \mathrm{kind}(v))$ and, for step sites, an
\emph{output contract} $C_V(v)$ (the type of the recorded result) and a set
of \emph{consumption contracts} $D_V(v)$ (the constraints that downstream
code in $V$ places on that result).
\end{definition}

\begin{definition}[Step graph]\label{def:graph}
The \emph{step graph} of $V$ is a directed graph
$G_V = (N_V \cup \{n_0, n_\bot\}, E_V)$ with distinguished entry $n_0$ and
exit $n_\bot$, where $(u,w) \in E_V$ iff some execution of $V$ performs
site $w$ immediately after site $u$. Edges out of branch points carry
predicates over previously recorded values. Structured concurrency
(fork--join groups such as \texttt{Promise.all}) gives $G_V$ a
series--parallel structure whose exact treatment, via trace equivalence,
is developed in \S\ref{sec:concurrency}.
\end{definition}

$G_V$ is computable by static analysis of the orchestration function: the
determinism contract confines control flow to a sandboxed subset of the
language, which is precisely what makes the extraction tractable.

\begin{definition}[Run, log, prefix, frontier]\label{def:run}
A \emph{run} $r$ of $V$ is identified by its immutable \emph{event log}
$L_r = (e_1, \dots, e_{k_r})$, where each event $e_i$ records a site
occurrence and, for completed steps, a payload $x_i$. The \emph{path
prefix} $\pref_r = (v_1, \dots, v_{k_r}) \in N_V^{*}$ is the sequence of
sites in $L_r$; the \emph{frontier} $\front_r$ is the site at which $r$ is
currently suspended (a pending step, timer, or hook).
\end{definition}

\begin{definition}[Replay]\label{def:replay}
\emph{Replay} of $L_r$ under version $V'$ executes $V'$ from its entry.
Each site invocation draws the next identifier from a single
deterministic per-run sequence shared by all site kinds (its
\emph{invocation ordinal}); recorded events carry the identifier drawn at
the original invocation, together with the site name. The log is consumed
in order: an event is matched to the pending invocation holding the same
identifier, subject to a name-equality check, and its recorded payload is
substituted. Replay \emph{diverges} if an event cannot be matched or the
name check fails; it \emph{succeeds} if all of $L_r$ is consumed without
divergence and without a runtime fault while consuming substituted
payloads; execution then proceeds live from the frontier. Site identity
is thus the pair (\emph{global invocation ordinal}, name)---strictly
finer than name-plus-path-position.
\end{definition}

\begin{remark}[Protocol conformance]\label{rem:conformance}
Definition~\ref{def:replay} and Assumptions~\ref{ass:det}--\ref{ass:log}
were checked against the reference implementation
(\texttt{@workflow/core}~\cite{wdk}): identifiers are ULIDs drawn from a
monotonic generator over the virtualized RNG, seeded by the run
identifier together with the workflow name and the run's \emph{recorded}
deployment identifier---replay-stable across deployments, which is
precisely what makes cross-version replay well-defined; the virtual clock
is initialized from the run identifier's embedded timestamp and advanced
to each consumed event's timestamp; unmatched events raise a hard
divergence error; end-of-log suspends at the frontier; retry events
(\texttt{step\_retrying}, repeated \texttt{step\_started}) are persisted,
grounding the re-execution telemetry of \S\ref{sec:severity}; and the
world surface exposes exactly the run and event listings that
Assumption~\ref{ass:log} and Algorithm~\ref{alg:wur} consume.

Two field observations from a live implementation of the analyzer
sharpen the telemetry lane. First, the event \emph{vocabulary} drifts
across releases: as of \texttt{workflow@5.0.0-beta.21} (confirmed on
beta.26), the default lazy inline step path emits no
\texttt{step\_created} event at all---the step name and input arrive on
the first \texttt{step\_started}---so an analyzer keyed to creation
events silently loses every step site. Implementations should key site
discovery on \texttt{step\_started} and treat the vocabulary as
versioned alongside the engine. Second, invocation ordinals are
\emph{directly recoverable from persisted data}: concurrent branches'
events are appended in scheduling order, but because all identifiers are
drawn from one monotonic ULID sequence, lexicographically sorting the
recorded correlation identifiers recovers the exact global invocation
order across steps, timers, and hooks---observed live on a
\texttt{Promise.all} whose events landed in log order $M, K, N$ while
their ULIDs ordered $K < M < N$, the true source-array order. The
analyzer never infers fork order from timestamps.

Race semantics are likewise implementation-confirmed, not assumed. The
timer consumer explicitly treats a timer completion as ``a
branch-deciding resolution the workflow may \texttt{Promise.race}
against a hook payload'' and orders competing deliveries
\emph{deterministically by event-log position} through delivery
barriers: a later-in-log hook payload is delivered only after an
earlier-in-log timer resolution and vice versa. The recorded log order
therefore \emph{is} the race outcome during replay, which is precisely
why Definition~\ref{def:indep} grants order-observable groups no
independent pairs. Losing branches are not cancelled: their consumers
remain registered, and a loser's completion appearing later in the log is
consumed normally. Explicit hook disposal is itself a persisted protocol
event (\texttt{hook\_disposed}), consumed by correlation as a terminal
marker on replay; at workflow completion, pending user hooks and timers
persist as pending invocations (only system hooks are implicitly
disposed), so a suspended run's frontier faithfully includes race losers.
\end{remark}

Figure~\ref{fig:identity} illustrates the identity discipline on a
suspended run: the same recorded prefix replayed against two candidate
versions, one insertion at the tail (harmless: every recorded event still
meets the consumer holding its identifier) and one at the head (fatal:
every subsequent ordinal shifts, so the first recorded event already
meets a consumer with a different name).

\begin{figure}[t]
\centering
\begin{tikzpicture}[node distance=2.6mm]
\node[ev] (e0) {0 \,\texttt{validate}};
\node[ev, below=of e0] (e1) {1 \,\texttt{charge}};
\node[ev, below=of e1] (e2) {2 \,\texttt{notify}};
\node[draw, densely dotted, below=of e2, minimum width=15mm,
      minimum height=5.2mm, font=\scriptsize] (fr) {\emph{frontier}};
\node[lab, above=1.5mm of e0] (rl) {recorded log $L_r$};
\node[tinylab, anchor=west] at ($(e0.east)+(0.5mm,0)$) {\,};
\node[site, left=14mm of e0] (a0) {0 \,\texttt{validate}};
\node[site, below=of a0] (a1) {1 \,\texttt{charge}};
\node[site, below=of a1] (a2) {2 \,\texttt{notify}};
\node[site, below=of a2, fill=black!10] (a3) {3 \,\texttt{audit}};
\node[lab, above=1.5mm of a0] {$V_2'$: insert at tail};
\draw[okline] (e0.west) -- (a0.east);
\draw[okline] (e1.west) -- (a1.east);
\draw[okline] (e2.west) -- (a2.east);
\draw[flow, densely dotted] (fr.west) -- (a3.east);
\node[tinylab, below=1.4mm of a3, align=center]
  {prefix valid; replay proceeds\\ live: $\Rb$ from demand terms only};
\node[site, right=14mm of e0, fill=black!10] (b0) {0 \,\texttt{audit}};
\node[site, below=of b0] (b1) {1 \,\texttt{validate}};
\node[site, below=of b1] (b2) {2 \,\texttt{charge}};
\node[site, below=of b2] (b3) {3 \,\texttt{notify}};
\node[lab, above=1.5mm of b0] {$V_2''$: insert at head};
\draw[badline] (e0.east) -- (b0.west)
      node[font=\small, pos=0.5, above=0.2mm] {$\times$};
\node[tinylab, below=1.4mm of b3, align=center]
  {event $(0,\texttt{validate})$ meets consumer\\
   $(0,\texttt{audit})$: name check fails,\\
   whole remaining prefix invalid, $\Rb=1$};
\end{tikzpicture}
\caption{Site identity is the pair (global invocation ordinal, name).
The same suspended run replayed against two candidates: a tail insertion
leaves every recorded identifier matched (left); a head insertion shifts
every subsequent ordinal, and replay diverges on the first event
(right). Any site kind---step, timer, or hook---shifts the sequence
equally.}
\label{fig:identity}
\end{figure}
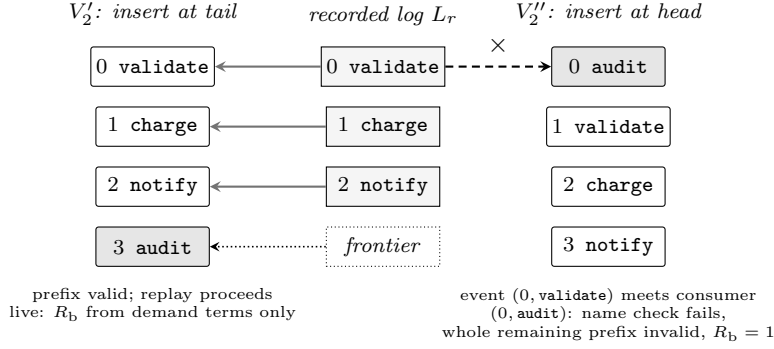

\subsection{Concurrency and Trace Equivalence}\label{sec:concurrency}

Structured concurrency complicates Definition~\ref{def:run}: within a
fork--join group the log records step completions in \emph{completion
order}, which is scheduling nondeterminism, not program semantics. Two
logs differing only by such permutations denote the same computation, and
prefix validity must not distinguish them.

\begin{definition}[Independence]\label{def:indep}
Sites $u, w$ are \emph{independent}, $(u,w) \in I$, iff they lie on
distinct branches of the same \emph{join-observed} parallel group---one
whose results the orchestrator can first observe at the join, as with
\texttt{Promise.all}. Independence quotients \emph{completion order in
the log only}: because site identity includes the global invocation
ordinal (Definition~\ref{def:replay}), the \emph{invocation order of
branches at the fork is protocol identity}, and reordering branches in
source---even within a join-observed group whose semantics are
order-free---is a protocol change in $\Delta_\prot$, not an $\equiv_I$
permutation. Groups with observable completion order
(\texttt{Promise.race}, \texttt{Promise.any}, incremental awaiting)
contribute no independent pairs at all: there, recorded order \emph{is}
semantics, and permutation sensitivity is correct rather than
spurious---a property the reference implementation enforces mechanically
by barriering branch-deciding deliveries in event-log order
(Remark~\ref{rem:conformance}).
\end{definition}

\begin{definition}[Trace equivalence]\label{def:trace}
$\equiv_I$ is the least congruence on site sequences that permutes
adjacent independent occurrences (Mazurkiewicz trace
equivalence~\cite{mazurkiewicz}). A run's log is one linearization of its
trace $[\pref_r]_{\equiv_I}$; all definitions over prefixes are lifted to
traces. In particular, prefix validity in Definition~\ref{def:rb} becomes:
$\pref_r \sqsubseteq_\match G_2$ iff some $\pref' \equiv_I \pref_r$ maps
under $\match$ to a valid path of $G_2$---equivalently, the run's
partially ordered prefix embeds into $G_2$'s series--parallel order.
\end{definition}

Branch correspondence within a parallel group is accordingly
\emph{positional}, not label-based: $\match$ must preserve the fork's
branch invocation order, and a permuted branch list is classified into
$\Delta_\prot$ even when a multiset comparison would succeed. The
recorded \emph{completions} of matched branches, by contrast, may appear
in any log order.

\begin{assumption}[Determinism contract]\label{ass:det}
Orchestration code is deterministic given the event log: nondeterministic
primitives are virtualized (e.g., pseudo-randomness seeded by the run
identifier; clock reads fixed to log-derived timestamps), and side effects
are confined to steps. Moreover, in join-observed parallel groups the
orchestration state after the join is a function of the \emph{multiset} of
branch results, independent of completion order.
\end{assumption}

\begin{assumption}[Complete telemetry]\label{ass:log}
The protocol persists, for every run past and present: its full event log,
including step payloads (or retrievable references to them), its status,
and the version that created it. The set of completed runs is denoted
$\mathcal{H}$ (history); the set of in-flight runs $\mathcal{R}$.
\end{assumption}

Assumption~\ref{ass:log} is not an idealization---it restates the
durability guarantee. It is the reason this domain admits exact backward
analysis where general distributed systems do not.

\begin{assumption}[Version availability]\label{ass:versions}
Both versions' code is available for comparison at decision time, in the
representation the runtime executes.
\end{assumption}

Like Assumption~\ref{ass:log}, this is less an idealization than a
restatement of how these engines already work. Pin-by-default versioning
\emph{requires} immutable, retained deployments---the engine cannot
replay a pinned run without the artifact that created it---so $V_1$'s
code is guaranteed to exist at $V_2$'s decision time by the same
mechanism the model serves; in git-based deploy flows, both versions
additionally resolve to commits. Two practicalities matter. The diff
must operate on the post-build representation, where module-qualified
site names are assigned---diffing raw source with different naming rules
than the runtime's yields a different protocol view; the build tooling
that assigns these names can equally enumerate them. And the diff's
scope is the full bundle closure: a dependency bump changes sites
without touching first-party source, so changed modules are identified
from the lockfile and the out-of-scope warning rule of
\S\ref{sec:algorithm} applies to anything the diff cannot see.

\section{Change Taxonomy and Incompatibility Functions}\label{sec:changes}

Let $V_1, V_2$ be the deployed and candidate versions, with step graphs
$G_1, G_2$.

\begin{definition}[Matching]\label{def:matching}
A \emph{matching} is a partial injective map $\match: N_1 \rightharpoonup
N_2$ such that $\ell(v) = \ell(\match(v))$, $\match$ preserves the
relative order of matched sites along every path of $G_1$, and---because
site identity includes the global invocation ordinal
(Definition~\ref{def:replay})---$\match$ preserves the invocation ordinal
of every site within the prefix region of any live run: an insertion or
deletion of \emph{any} site kind (step, timer, or hook) ahead of recorded
events shifts all subsequent ordinals and invalidates them wholesale.
Sites are paired by label and structural position. Two properties of the
label deserve emphasis, both engine-real. First, label granularity is
engine-defined and in Workflow SDK is \emph{module-qualified}
(\texttt{step//./workflows/e2e//add}): moving a step function to another
file changes every affected site's name, and because the runtime's own
replay name check compares qualified names, the resulting fleet-wide
$\Rb = 1$ is a \emph{true} verdict for a semantically identical
refactor---the engine genuinely diverges. A path-normalized matching
(pairing on the trailing function segment) may refine the
\emph{diagnosis} from delete-plus-insert to rename, preserving forward
telemetry lineage and letting policy propose a rename-map migration, but
it must never weaken the backward verdict, which the engine's identity
discipline fixes. Second, the \emph{call-site fingerprint} (argument
data-flow shape at the invocation) is a soundness requirement wherever a
label repeats, not an ambiguity tiebreak: swapping two calls to the same
step function---\texttt{add(a,3)} and \texttt{add(b,5)}---leaves the
(ordinal, name) sequence unchanged, so replay \emph{accepts} the log and
substitutes each recorded result into the other call site, the silent
result-mismatch failure of class F2; a one-line refactor produces it on
the reference repository's own corpus, where repeated invocations of one
step (loops, fan-out) are common. Matching must therefore distinguish
same-labeled sites by fingerprint, classifying a fingerprint-crossing
swap into $\Delta_\intf$ with $\rho = 1$ at the affected sites even
though the protocol view is unchanged. Branch predicates guarding
matched edges must be syntactically equivalent after normalization;
otherwise the edge is treated as changed.
\end{definition}

Syntactic predicate comparison is conservative by design: semantic
equivalence is undecidable, and Theorem~\ref{thm:sound} requires only that
matching never equates sites or predicates that differ observably.

\begin{definition}[Change set]\label{def:delta}
The \emph{change set} $\Delta = \Delta_\prot \cup \Delta_\intf \cup
\Delta_\migr$ comprises:
\begin{itemize}[leftmargin=2em]
  \item $\Delta_\prot$ (\emph{protocol}): sites of $N_1$ outside
  $\mathrm{dom}(\match)$, sites of $N_2$ outside $\mathrm{ran}(\match)$,
  order violations, and changed branch predicates. These alter the language
  of site sequences that $G_2$ accepts.
  \item $\Delta_\intf$ (\emph{interface}): matched sites $v$ with
  $D_2(\match(v)) \not\sqsupseteq C_1(v)$ in the compatibility lattice
  of Definition~\ref{def:lattice}---the new consumers demand more than the
  old producer guaranteed.
  \item $\Delta_\migr$ (\emph{migration}): changes to the serialized shape
  of workflow-scope state that crosses suspension boundaries (captured
  closures, hook payload contracts).
\end{itemize}
\end{definition}

\begin{definition}[Compatibility lattice]\label{def:lattice}
For record-like payload types, define the partial order $\sqsubseteq$
generated by: adding an optional field, widening a field's type, and
relaxing a constraint are $\sqsubseteq$-increasing (compatible); removing a
field, renaming a field, adding a required field, and narrowing a type are
incomparable or decreasing (potentially incompatible). This mirrors
reader/writer compatibility in schema-evolution systems~\cite{kleppmann}.
\end{definition}

The lattice is defined over \emph{hydrated} payloads, and that word is
doing work: the reference implementation persists step results in a
structured serialization (a flattened value array with numeric
back-references and typed wrappers for \texttt{Date}, \texttt{Map},
\texttt{Set}, and cycles), so field presence is well-defined only after
hydration, and an implementation computing presence frequencies over raw
serialized bytes will be wrong. Typed wrappers sit outside the record
lattice entirely: a demand reaching \emph{into} a \texttt{Map} or
\texttt{Date} internal receives verdict \textsc{unk} and the maximal-%
demand treatment of \S\ref{sec:demand}, since the lattice has nothing
sound to say about their evolution. Both the exact and distributional
regimes of \S\ref{sec:estimation} therefore include a hydration step;
its cost is linear in payload size and is included in the measured
per-run figures of \S\ref{sec:evaluation}.

\begin{definition}[Incompatibility functions]\label{def:rho}
Each changed site is scored by an \emph{incompatibility function}
$\rho: N_1 \to [0,1]$, the probability that a recorded occurrence of the
site is incompatible with $V_2$:
\begin{align}
\rho_\prot(v) &\in \{0,1\}
  && \text{deterministic: membership in } \Delta_\prot, \label{eq:rhoP}\\
\rho_\intf(v) &= \Prob_{x \sim \Phi_v}\!\bigl[\,x \nvDash D_2(\match(v))\,\bigr]
  && \text{payload-distribution mass violating $V_2$'s demand}, \label{eq:rhoI}\\
\rho_\migr(v) &\in \{0,1\}
  && \text{lattice comparison of serialized-state shapes}, \label{eq:rhoM}
\end{align}
where $\Phi_v$ is the empirical distribution of payloads recorded at $v$
across $\mathcal{H} \cup \mathcal{R}$, and $x \nvDash D$ denotes structural
violation of contract $D$.
\end{definition}

Equation~\eqref{eq:rhoI} is the pivotal substitution that removes the need
for dry-run execution. Whether $V_2$ can consume $V_1$'s recorded outputs
is not a property of code paths that must be exercised; it is a
\emph{measure-theoretic} property of the observed payload distribution
against a static contract. If a new consumption site destructures field
$x$, then $\rho_\intf = 1 - \mathrm{freq}_v(x)$; if it narrows a numeric
domain, $\rho_\intf$ is the observed mass outside that domain. Two regimes
are available:

\begin{itemize}[leftmargin=2em]
  \item \textbf{Distributional} (aggregate): $\Phi_v$ summarized as
  field-presence frequencies and value-domain histograms; $\rho_\intf$ is a
  fleet-level rate. Cheap; supports pre-deployment scoring.
  \item \textbf{Exact} (per-run): the actual payload $x_{r,v}$ is validated
  against $D_2$, collapsing $\rho_\intf$ to an indicator per run. Still
  purely static---data validated against a schema, no code executed.
\end{itemize}

\subsection{Demand Inference and the Differential-Demand Lemma}\label{sec:demand}

The consumption contract $D_2(v)$ is obtained from static types where the
program supplies them and, in untyped code, by \emph{access-path abstract
interpretation} of the consumer slice of $v$: the set of member accesses,
destructuring patterns, comparison and arithmetic domains, and
callee-signature flows reachable from the binding of $v$'s result.
Contract verdicts are three-valued,
$x \vDash D \in \{\textsc{sat}, \textsc{vio}, \textsc{unk}\}$: a demand
the analysis cannot resolve (computed property access, reflective use)
widens \emph{only the affected access path} to \textsc{unk}. An
\textsc{unk} verdict never enters Equation~\eqref{eq:rb} as a certainty in
either direction; it contributes a weak-prior Beta term for $\rho_\intf$
with low confidence, which under the pooling rule of \S\ref{sec:pooling}
routes the affected runs toward \textsc{review} rather than silently
toward \textsc{migrate} or \textsc{pin}.

Crucially, the inference burden is bounded by the diff, not the program:

\begin{lemma}[Differential demand]\label{lem:diffdemand}
If the consumer slice of site $v$ is unchanged under $\match$
(syntactically identical after normalization, with unchanged reachable
callees), then setting $\rho_\intf(v) = 0$ preserves the soundness of
Theorem~\ref{thm:sound}.
\end{lemma}

\begin{proof}[Proof sketch]
Replay under $V_2$ presents the recorded payload $x_{r,v}$ to code
identical to the code that consumed it under $V_1$. History witnesses that
this consumption did not fault---the run progressed past $v$. By
Assumption~\ref{ass:det}, the unchanged slice's behavior on identical
inputs is identical, so consumption under $V_2$ cannot fault either. No
contract is validated: certification is by determinism and the historical
witness, bypassing $D_2$ entirely.
\end{proof}

Contract inference (and Assumption~\ref{ass:contract}) is therefore
required only at \emph{changed} consumer slices---typically a handful of
access paths per release---rather than across the codebase. Historical
success is itself compatibility evidence wherever code did not change; the
analyzer spends its precision budget exactly where the diff spends its
risk. Per-site incompatibilities across axes are pooled into a single
$\rho(v)$ by the confidence-weighted operator of \S\ref{sec:pooling}.

\paragraph{Demand satisfiability and the invocation gate.}
One class of change escapes both the differential-demand lemma and
per-site contract validation: the consumer slice is untouched, but its
\emph{producer} is deleted or moved after the consumer, so the demand
edge is no longer satisfiable in $V_2$'s structure at all. The analysis
therefore includes a static satisfiability check over $V_2$: for every
demand edge $(c, u)$, the producer $u$ must lie in the
certainly-preceding scope of $c$ (preceding at top level, in the same
parallel branch, or in the same conditional arm, together with everything
preceding the enclosing construct). An unsatisfiable demand faults
whenever its consumer is invoked, so the per-run question reduces to
whether replay invokes $c$; because the protocol may have changed, that
gate must be evaluated against $V_2$'s structure, not the run's recorded
frontier: replay invokes $c$ once every site certainly preceding $c$ in
$V_2$ that $V_1$ can produce has completed in the prefix. Both the rule
and the gate were forced by the real-workflow corpus of
\S\ref{sec:evaluation}, where producer deletions and demand-order
inversions arise naturally (hook results consumed by later steps) and
where a $V_1$-relative frontier misses consumers that a reordered $V_2$
invokes first.

\paragraph{Dynamic access and inference limits.}
Access-path abstract interpretation reads statically visible member
accesses; JavaScript offers several ways to hide them: computed access
(\texttt{payload[fieldName]}), reflective access
(\texttt{Reflect.get}), deep-path helpers (\texttt{lodash.get}),
re-serialization boundaries (\texttt{JSON.parse} of a stringified
payload), and aliasing through intermediate collections. The design
response is uniform: any consumer slice whose demand cannot be bounded
statically receives verdict \textsc{unk} on the \emph{entire} payload
of each producer it can reach, which the risk model treats as maximal
demand---conservative, never unsound, and it degrades precision only for
the changed slices the differential-demand lemma routes to inference in
the first place. Downstream, an \textsc{unk} verdict is not a silent
worst case: it enters the pooling of \S\ref{sec:pooling} as a
maximal-demand factor with minimal confidence, so its wide Beta
posterior inflates the credible interval, and the interval-based policy
of \S\ref{sec:policy} routes the affected runs to \textsc{review}
rather than to silent migration or blanket pinning---the class whose
diagnosis names the offending consumer slice. The confidence attached to an \textsc{unk} factor is an
operator-tunable prior with two meaningful endpoints: uninformative
(wide posterior, routes to \textsc{review}---the default, since
unknown is not broken) and degenerate-pessimistic (certain violation,
routes to \textsc{pin}); one tested implementation of the analyzer
ships the pessimistic endpoint, which is strictly safer and strictly
more expensive, and both are conformant. Two things bound the
damage in practice: unchanged slices never need inference at all (their
certification is historical, \S\ref{sec:demand}), and \textsc{unk}
verdicts are exactly the places where a one-line developer annotation
(an explicit field list) restores precision; surfacing them in the
per-release verdict explanation makes the annotation a review-time
action rather than archaeology. The evaluation of \S\ref{sec:evaluation} deliberately
evaluates the risk model \emph{given} demands rather than the inference
front-end, which is language- and engine-specific; an empirical study of
inference precision on production JavaScript is future work and is
flagged as such rather than assumed away.

\section{The Risk Model}\label{sec:risk}

Upgrade risk for a run decomposes into a \emph{backward} component---can
the recorded history rehydrate under $V_2$?---and a \emph{forward}
component---will the remainder of the run traverse changed code?
The components differ in kind: backward incompatibility is
\emph{corruption} (state loss, divergence, stranding), forward exposure is
\emph{behavioral} (a fresh execution of intentionally changed code).

\subsection{Backward (rehydration) risk}

\begin{definition}[Backward risk]\label{def:rb}
For run $r$ with prefix $\pref_r = (v_1,\dots,v_k)$ and frontier site
$v_{k+1}$ (the pending invocation recorded by the protocol),
\begin{equation}\label{eq:rb}
\Rb(r) \;=\; 1 \;-\;
\ind{\pref_r \sqsubseteq_\match G_2}\;
\prod_{i=1}^{k+1}\bigl(1-\rho_\intf(v_i)\bigr)
\prod_{i=1}^{k}\bigl(1-\rho_\migr(v_i)\bigr),
\end{equation}
where $\pref_r \sqsubseteq_\match G_2$ (\emph{prefix validity}) holds iff
$(\match(v_1),\dots,\match(v_k))$ is defined and is a path from the entry
of $G_2$ whose guarding predicates are matched-equivalent. The indicator
is read as a hard gate: if the trace-equivalent prefix is not valid under
$\match$, then $\Rb(r) = 1$ (certain replay failure) regardless of the
survival products; when it is valid, the two products compose per-site
survival probabilities over interface and migration factors, so
$\Rb(r) = 0$ exactly when every factor is zero.
\end{definition}

Two boundary conditions matter, and both were surfaced by the synthetic
evaluation of \S\ref{sec:evaluation} before being incorporated here.
First, the interface product ranges over $k{+}1$ sites, not $k$: replay of
a complete prefix immediately \emph{invokes} the frontier site, and
constructing its arguments consumes recorded state, so a demand
introduced at the frontier faults during replay, before any live
execution. Second, the store against which a site's demand is validated
is the \emph{invocation-time} store, not the end-of-prefix store: within
a concurrency group, a producer's completion may appear later in the log
than a sibling consumer's invocation, so a changed consumer that reaches
across branches can fault even though the demanded payload exists
somewhere in the prefix. Both are reconstructible from the log alone,
since the protocol records pending invocations and completion order.

The structure is a survival product: rehydration requires every recorded
event to be consumable, so a single incompatible site anywhere in history
is fatal. Prefix validity captures failure mode (F1); the interface product
captures (F2); the migration product captures serialized-state
deserialization. In the exact regime the interface factors are indicators
and $\Rb(r) \in \{0,1\}$: backward risk is then not an estimate but a
\emph{verdict}. Because $\pref_r$ is fully known
(Assumption~\ref{ass:log}), no probabilistic approximation of history is
ever necessary---runs may be bucketed by prefix hash so the computation is
shared across runs with identical histories.

\subsection{Forward (behavioral) risk}

Forward risk requires a model of where the run will go. We estimate an
absorbing Markov chain on $G_1$'s sites from the complete path histories in
$\mathcal{H}$.

\begin{definition}[Empirical control-flow chain]\label{def:chain}
Let $c(u,w)$ be the number of observed transitions $u \to w$ across all
paths in $\mathcal{H}$. The transition matrix $P$ has rows
$P(u,\cdot) \sim \mathrm{Dirichlet}(\bm{\alpha}_u + \bm{c}_u)$ in the
Bayesian treatment (\S\ref{sec:estimation}); its posterior mean is used for
point estimates. The exit $n_\bot$ is absorbing.
\end{definition}

\begin{definition}[Hitting probabilities]\label{def:hit}
For a changed site $v \in \Delta$, the hitting probability
$h_v: N_1 \to [0,1]$ is the minimal solution of
\begin{equation}\label{eq:hit}
h_v(v) = 1, \qquad h_v(n_\bot) = 0, \qquad
h_v(u) = \sum_{w} P(u,w)\, h_v(w) \quad (u \neq v, n_\bot),
\end{equation}
the standard first-passage system~\cite{norris}. Writing $Q$ for the
transient submatrix, all $h_v$ are obtained from one factorization of
$(I-Q)$.
\end{definition}

\begin{definition}[Forward risk]\label{def:rf}
For a run suspended at frontier $\front_r$,
\begin{equation}\label{eq:rf}
\Rf(r) \;=\; 1 - \prod_{v \in \Delta}
  \bigl(1 - \sev(v)\, h_v(\front_r)\bigr),
\end{equation}
where $\sev(v) \in [0,1]$ is the severity assigned to the change at $v$.
\end{definition}

The noisy-or composition treats changed sites as independent hazard
sources; correlated changes (one logical edit touching several sites)
should be grouped into a single hazard before composition. Severity is
not assigned by judgment but \emph{estimated} from telemetry and
deployment history (\S\ref{sec:severity}); domain knowledge enters only
as priors that data progressively displaces.

\subsection{Aggregate score}

\begin{definition}[Workflow Upgrade Risk]\label{def:wur}
Per run and over the fleet,
\begin{equation}\label{eq:wur}
\WUR(r) \;=\; 1 - \bigl(1-\Rb(r)\bigr)\bigl(1-\alpha\,\Rf(r)\bigr),
\qquad
\WUR \;=\; \sum_{r \in \mathcal{R}} \omega_r\, \WUR(r),
\end{equation}
with $\alpha \in [0,1]$ discounting behavioral relative to corruption risk
and weights $\omega_r$ ($\sum \omega_r = 1$) defaulting to uniform but
admitting business criticality.
\end{definition}

Equivalently, aggregating over the frontier distribution
$\pi(\front) = \sum_r \omega_r \ind{\front_r = \front}$ when per-run
exactness is not needed:
$\WUR = \sum_{\front} \pi(\front)\,[\,1-(1-\bar\Rb(\front))(1-\alpha
\Rf(\front))\,]$, where $\bar\Rb(\front)$ averages backward risk over the
prefixes observed at $\front$.

\section{Estimation from Telemetry}\label{sec:estimation}

All quantities in \S\ref{sec:risk} are estimated from persisted data; this
section makes the uncertainty explicit so that $\WUR$ is reported as a
posterior, not a point.

\subsection{Transition posterior}

With Dirichlet prior $\bm{\alpha}_u$ (default: symmetric, concentration
$1$ over $G_1$-feasible successors) and counts $\bm{c}_u$, the posterior
$P(u,\cdot) \sim \mathrm{Dirichlet}(\bm{\alpha}_u + \bm{c}_u)$ has mean
$\hat P(u,w) = (\alpha_{uw} + c_{uw}) / (\alpha_{u\cdot} + c_{u\cdot})$.
A branch observed a dozen times carries a visibly wide posterior; a branch
never observed reverts to the prior---and is flagged
(\S\ref{sec:policy}).

\subsection{Interface posterior}

For site $v$ with $n_v$ recorded payloads of which $m_v$ violate
$D_2(\match(v))$, $\rho_\intf(v) \sim \mathrm{Beta}(m_v + a,\, n_v - m_v +
b)$ with a weak prior $(a,b)$. The distributional regime uses this
posterior; the exact regime bypasses it per run.

\subsection{Cross-axis pooling}\label{sec:pooling}

Per-site axis scores $\rho_\prot, \rho_\intf, \rho_\migr$ with confidences
$w_\prot, w_\intf, w_\migr$ (posterior precisions; deterministic axes get
$w = \infty$, dominating) are pooled in log-odds space:
\begin{equation}\label{eq:pool}
\operatorname{logit} \rho(v) \;=\;
\frac{\sum_a w_a \operatorname{logit} \rho_a(v)}{\sum_a w_a},
\end{equation}
with the convention that any axis at $\rho_a = 1$ with certainty forces
$\rho(v) = 1$. Confidence-weighted log-odds pooling has the two properties
the setting demands: independent weak signals compound, and a single
certain signal dominates regardless of how many uncertain signals disagree.

\subsection{Severity Estimation}\label{sec:severity}

Severity decomposes into three estimable factors, pooled by
Equation~\eqref{eq:pool}; each carries a posterior, so severity
uncertainty propagates into the $\WUR$ interval like every other
estimate.

\paragraph{Re-execution tolerance.} Durable engines retry failed steps
automatically, so the event logs already contain \emph{natural
re-execution experiments}: every recorded retry of step $s$ that
subsequently completed, in a run that finished without an operator-visible
incident, is a Bernoulli trial witnessing that an extra execution of $s$
is tolerated. With $t_s$ such benign re-executions and $i_s$ incident-linked
ones, $\kappa(s) \sim \mathrm{Beta}(i_s + a,\, t_s + b)$ estimates the
probability that re-executing $s$ is harmful---an \emph{empirical
idempotency measure} requiring no annotation. Idempotency annotations,
where present, enter as informative priors rather than overrides.

\paragraph{Blast radius.} The static count of downstream sites consuming
$v$'s output (normalized by graph size) weights changes whose effects
propagate; it is exact, carrying no posterior.

\paragraph{Class calibration.} Past releases are natural experiments at
the change-class level. The operator's own deployment history yields
per-class $(\text{incident} \mid \text{exposure})$ records, from which
empirical-Bayes Beta posteriors are estimated with shrinkage toward a
cross-fleet prior. A team that has shipped fifty releases has, in effect,
already labeled a severity training set; a team with none inherits the
prior and sees correspondingly wider $\WUR$ intervals---the cold-start
cost is reported, not hidden.

\subsection{Propagation}

$\WUR$'s posterior is obtained by Monte Carlo: draw
$P^{(s)}, \rho^{(s)}$ from their posteriors, recompute
\eqref{eq:hit}--\eqref{eq:wur}, and report the posterior mean with a
$90\%$ credible interval. Because \eqref{eq:hit} is a sparse linear solve
on a graph of tens of nodes, hundreds of draws cost milliseconds. A wide
interval is itself an actionable output: \emph{insufficient telemetry to
certify this migration}.

\section{Theoretical Properties}\label{sec:properties}

\begin{assumption}[Sound matching]\label{ass:match}
$\match$ equates sites only if their labels are identical and equates
guard predicates only if they are semantically equivalent (guaranteed
conservatively by syntactic normalization: syntactic difference $\Rightarrow$
treated as changed).
\end{assumption}

\begin{assumption}[Sound contracts at changed slices]\label{ass:contract}
At sites whose consumer slice is changed under $\match$, the consumption
contracts $D_2$ over-approximate $V_2$'s actual runtime demands on
substituted payloads: if replaying code faults or misbehaves on payload
$x$, then $x \nvDash D_2$ (\textsc{vio} or \textsc{unk}). Unchanged
consumer slices require no contract by Lemma~\ref{lem:diffdemand}.
\end{assumption}

\begin{theorem}[Certification]\label{thm:sound}
Under Assumptions~\ref{ass:det}--\ref{ass:contract}, if $\Rb(r) = 0$ in the
exact regime, then replay of $L_r$ under $V_2$ succeeds and suspends at the
site $\match(\front_r)$ with the same substituted values as replay under
$V_1$.
\end{theorem}

\begin{proof}[Proof sketch]
Induction over the series--parallel decomposition of the prefix.
Invariant: after consuming a downward-closed set of events,
$V_2$'s execution state coincides (up to $\match$) with $V_1$'s replay
state on the same events. \emph{Sequential composition:} by the invariant,
$V_2$ evaluates the same guard predicates on the same substituted values
as $V_1$; prefix validity (now modulo $\equiv_I$,
Definition~\ref{def:trace}) plus Assumption~\ref{ass:match} gives
identical branch outcomes, so $V_2$ reaches the matched site of the next
event---no divergence. $\Rb(r)=0$ gives payload compatibility at that
site: by Lemma~\ref{lem:diffdemand} if the consumer slice is unchanged,
and by Assumption~\ref{ass:contract} otherwise; so substitution does not
fault, and by Assumption~\ref{ass:det} the post-substitution state is a
deterministic function of substituted values, re-establishing the
invariant. \emph{Join-observed parallel composition:} the invariant holds
per branch, since branches are mutually unobservable before the join;
matched branches, invoked in the same fork order, regenerate identical
invocation ordinals, so every recorded completion has a registered
consumer regardless of the log's completion order---exactly the
permutations $\equiv_I$ quotients away, and exactly the tolerance the
reference implementation exhibits (Remark~\ref{rem:conformance}). The
post-join state is by Assumption~\ref{ass:det} a function of the multiset
of branch results, which the matched branches reproduce, so the invariant
is re-established at the join.
\emph{Order-observed groups} contribute no independent pairs, so their
events are handled by the sequential case. After all events are consumed,
execution suspends at $\match(\front_r)$.
\end{proof}

\begin{remark}
The converse fails, deliberately: $\Rb(r) = 1$ may be a false positive
induced by conservative syntactic matching (e.g., a semantically neutral
refactor of a guard). The model trades completeness for soundness; the
policy layer (\S\ref{sec:policy}) routes such runs to \emph{pin}, which is
safe.
\end{remark}

\begin{theorem}[Canonical-form decision for concurrent prefixes]\label{thm:canon}
For series--parallel $G_1, G_2$, prefix validity modulo $\equiv_I$ is
decidable in $O(|\pref_r| \log |\pref_r|)$ time.
\end{theorem}

\begin{proof}[Proof sketch]
Series--parallel decompositions are unique~\cite{vtl}, so both the
recorded prefix (a downward-closed sub-pomset of $G_1$'s sp-order) and the
candidate region of $G_2$ admit canonical decomposition trees. Define the
normal form of a prefix recursively: at sequential nodes, concatenate
children's normal forms; at join-observed parallel nodes, order the
branches' normal forms by \emph{branch invocation ordinal}---the
canonical key both versions share, since branch order is protocol
identity (Definition~\ref{def:indep}); in the reference implementation
this key is materialized, not reconstructed, since a lexicographic sort
of the recorded correlation identifiers yields it directly
(Remark~\ref{rem:conformance}). Two prefixes are
$\equiv_I$-equivalent iff their normal forms are equal (uniqueness of
normal forms for traces over sp-independence alphabets), and the prefix
embeds in $G_2$ iff its normal form matches the normal form of the
$\match$-image region, checked by one simultaneous traversal. Sorting
completions into branch order dominates the cost.
\end{proof}

\begin{corollary}
$\Rb$ remains exactly computable in the presence of \texttt{Promise.all}
concurrency, at unchanged asymptotic cost: the $O(b \cdot L)$ bound of
Proposition~\ref{prop:complexity} gains only the per-prefix
$O(L \log L)$ normalization.
\end{corollary}

\begin{proposition}[Exactness of the backward term]\label{prop:exact}
In the exact regime, $\Rb$ involves no estimation: it is a computable
function of $(L_r, G_1, G_2, \match)$. In particular runs with identical
prefix hashes have identical $\Rb$, and the fleet-level backward risk is
the exact fraction (under $\omega$) of rehydration-incompatible runs.
\end{proposition}

\begin{proposition}[Monotonicity]\label{prop:mono}
$\WUR$ is non-decreasing in each $\rho(v)$, each $\sev(v)$, each
$h_v(\front)$, and under enlargement of $\Delta$. Consequently splitting a
release into smaller change sets never increases, and generically
decreases, per-release risk---recovering, as a theorem of the model, the
operational folklore that small frequent deploys are safer.
\end{proposition}

\begin{proposition}[Complexity]\label{prop:complexity}
With $n = |N_1|$, $m$ in-flight runs bucketed into $b \le m$ distinct
prefixes of maximum length $L$, and $d = |\Delta|$: graph extraction and
diff are linear in program size with an $O(n^2)$ matching resolution;
backward risk costs $O(b \cdot L)$ contract validations; forward risk
costs one $O(n^3)$ worst-case (sparse, in practice near-linear)
factorization of $I - Q$ plus $O(d\,n)$ back-substitutions; aggregation is
$O(m + d\,b)$. Monte Carlo propagation multiplies the forward and
aggregation terms by the draw count $S$.
\end{proposition}

Workflow step graphs are small ($n$ in the tens); the binding cost is
payload validation, which is embarrassingly parallel and incremental
(memoizable by prefix hash across candidate releases).

\section{Algorithm}\label{sec:algorithm}

Figure~\ref{fig:pipeline} summarizes the functional design: a static lane
that never executes anything, a telemetry lane that never inspects
source, and a per-run join whose outputs feed the fleet-level coupling
stage. Algorithm~\ref{alg:wur} gives the same structure operationally.

\begin{figure}[t]
\centering
\begin{tikzpicture}[node distance=4.5mm and 8mm]
\node[stage, minimum width=52mm] (src) {$V_1,\,V_2$ source};
\node[stage, right=of src, minimum width=62mm] (world)
  {protocol telemetry: \texttt{runs.list}, \texttt{events.list},\\
   retry events, deployment history\ \ (\S\ref{sec:estimation})};
\node[stage, below=of src, minimum width=52mm] (klass)
  {matching $\match$ and canonical diff (\S\ref{sec:model});\\
   change classes $\Delta_\prot, \Delta_\intf, \Delta_\migr$;\\
   demand satisfiability (\S\ref{sec:changes},\,\ref{sec:demand})};
\node[stage, below=of world, minimum width=62mm] (dyn)
  {Markov chain, hitting probabilities, severity $\sigma$;\\
   per-run prefix, frontier, invocation-time stores};
\node[stage, below=6mm of $(klass.south)!0.5!(dyn.south)$,
      minimum width=64mm] (risk)
  {per run $r$: backward $\Rb(r)$ (exact or distributional),\\
   forward $\Rf(r)$ (hitting $\times$ severity);
   pooled $\WUR_r$ posterior\ \ (\S\ref{sec:risk})};
\node[stage, below=6mm of risk, xshift=-40mm, minimum width=38mm] (coup)
  {coupling graph;\\ contagion fixpoint (\S\ref{sec:coupling})};
\node[stage, right=of coup, minimum width=22mm] (cut)
  {min-cut\\ partition};
\node[stage, right=of cut, minimum width=30mm] (out)
  {\textsc{migrate} / \textsc{review} / \textsc{pin}\\ per run};
\draw[flow] (src) -- (klass);
\draw[flow] (world) -- (dyn);
\draw[flow] (klass.south) -- ($(risk.north)+(-25mm,0)$);
\draw[flow] (dyn.south) -- ($(risk.north)+(25mm,0)$);
\draw[flow] (risk.south) -- ($(risk.south)+(0,-3mm)$) -| (coup.north);
\draw[flow] (coup) -- (cut);
\draw[flow] (cut) -- (out);
\node[lab, anchor=south west] at ($(src.north west)+(0,0.8mm)$)
  {static analysis (no execution)};
\node[lab, anchor=south east] at ($(world.north east)+(0,0.8mm)$)
  {telemetry (no source inspection)};
\end{tikzpicture}
\caption{Functional design. The static lane classifies the diff and
checks demand satisfiability without executing anything; the telemetry
lane reconstructs per-run state and estimates dynamics from the protocol
surface without inspecting source. They meet only in the per-run risk
terms, whose posteriors drive the fleet-level contagion and partition
stages. No stage replays a log against $V_2$.}
\label{fig:pipeline}
\end{figure}
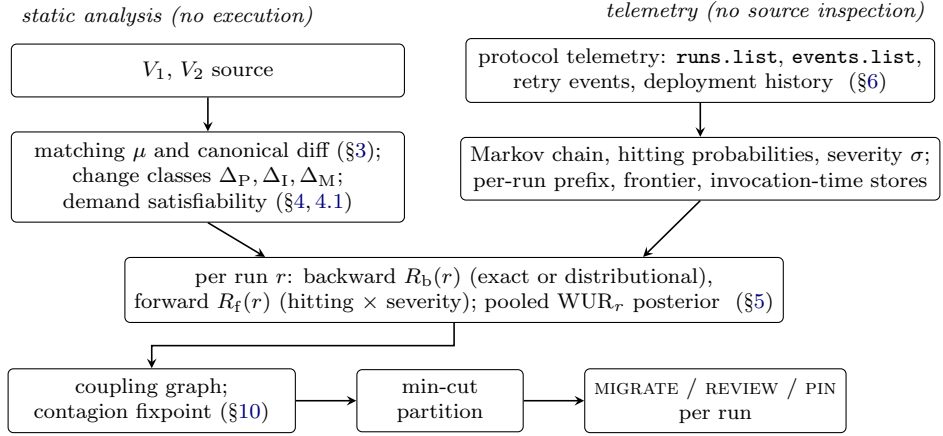

\begin{algorithm}[t]
\caption{$\WUR$: telemetry-driven upgrade risk}\label{alg:wur}
\begin{algorithmic}[1]
\Require versions $V_1, V_2$; world handle $\mathcal{W}$; weights
$\omega$; discount $\alpha$; draws $S$
\State $G_1, G_2 \gets \textsc{ExtractStepGraph}(V_1), \textsc{ExtractStepGraph}(V_2)$
\State $\match \gets \textsc{Match}(G_1, G_2)$;\quad
$\Delta \gets \textsc{Classify}(G_1, G_2, \match)$
  \Comment{axes $\prot,\intf,\migr$}
\State $\mathcal{H}, \mathcal{R} \gets \mathcal{W}.\textsc{runs}$;\quad
$\{L_r\} \gets \mathcal{W}.\textsc{events}$
\State $\bm{c} \gets \textsc{TransitionCounts}(\mathcal{H})$;\quad
$\Phi \gets \textsc{PayloadStats}(\mathcal{H} \cup \mathcal{R})$
\For{$v \in \Delta$}
  \State compute $\rho_\prot, \rho_\migr$ statically; posterior of
  $\rho_\intf$ from $\Phi_v$ vs.\ $D_2(\match(v))$ \Comment{Eqs.~\eqref{eq:rhoP}--\eqref{eq:rhoM}}
  \State $\rho(v) \gets \textsc{PoolLogOdds}(\rho_\prot, \rho_\intf, \rho_\migr)$ \Comment{Eq.~\eqref{eq:pool}}
\EndFor
\For{each distinct prefix bucket $\pref$ in $\mathcal{R}$}
  \State $\Rb(\pref) \gets$ Eq.~\eqref{eq:rb}
  \Comment{exact; validate stored payloads against $D_2$}
\EndFor
\For{$s = 1..S$} \Comment{posterior propagation}
  \State draw $P^{(s)} \sim \mathrm{Dirichlet}(\bm{\alpha} + \bm{c})$,
  $\rho^{(s)}, \tau^{(s)}, \lambda^{(s)}$ from their posteriors
  \State solve Eq.~\eqref{eq:hit} once for $(I - Q^{(s)})$; read off
  $h_v^{(s)}(\cdot)$ for all $v \in \Delta$
  \State $\WUR^{(s)} \gets$ Eqs.~\eqref{eq:rf}--\eqref{eq:wur}
  \State $K \gets \textsc{CouplingGraph}(\mathcal{H}, \mathcal{R})$;
  \ $q^{*(s)} \gets \lim_k F^k(\bm{0})$ \Comment{Eq.~\eqref{eq:fix}}
  \State $a^{(s)} \gets \textsc{MinCut}\bigl(\psi^{(s)}, \pi^{(s)}\bigr)$
  \Comment{Thm.~\ref{thm:cut}; contract $\pi{=}\infty$ edges first}
\EndFor
\State \Return posterior mean and $90\%$ interval of $\WUR$
(coupled: $\sum_r \omega_r q_r^{*}$); stable-core partition with
label-instability \textsc{review} set; per-change annotations
\end{algorithmic}
\end{algorithm}

Algorithm~\ref{alg:wur} touches the world only through the protocol surface
(run and event listings); it executes no workflow code and requires no
sandbox. Its outputs are (i) the fleet score with credible interval,
(ii) a per-run verdict vector, and (iii) per-change annotations---which
edit induces which risk, at which sites, affecting how many runs---the
explanation artifact an operator sees at deploy time.

\paragraph{Implementation notes.}
The reference implementation used for \S\ref{sec:evaluation} (a
single-file Python analyzer consuming only run listings and event logs)
realizes Algorithm~\ref{alg:wur} directly and hits the costs of
Proposition~\ref{prop:complexity} without tuning; its measured numbers
appear in \S\ref{sec:evaluation}. One production optimization is worth
recording because it changes the constant dramatically: $\Rb$ depends
only on the diff and the run's (prefix content, frontier) pair, so runs
sharing a frontier and trace-equivalent prefix share a verdict. Hashing
the canonical normal form of each prefix (\S\ref{sec:properties}) and
computing once per equivalence class turns per-run work into per-class
work; production fleets concentrate mass on few frontiers (runs parked
at the same timer or hook), so the class count is typically orders of
magnitude below the run count. The per-run exact payload checks that do
vary within a class (recorded field presence) reduce to a bitmask over
the demanded fields, batched trivially. One scoping rule is
soundness-relevant rather than an optimization: the matching of
Definition~\ref{def:matching} is defined over whole versions, so an
analyzer whose diff covers less than every module contributing
sites---a single-file diff cannot observe module-qualified renames, for
instance---must surface out-of-scope sites as explicit warnings, never
as silent certifications.

\section{Migration Policy}\label{sec:policy}

The verdict partition follows from thresholds
$0 < \tau_{\mathrm{lo}} < \tau_{\mathrm{hi}} < 1$ on per-run scores:

\begin{center}
\begin{tabular}{@{}lll@{}}
\toprule
Condition & Class & Action \\
\midrule
$\Rb(r) = 0 \;\wedge\; \alpha \Rf(r) < \tau_{\mathrm{lo}}$ & \textsc{migrate} & route run to $V_2$ \\
$\Rb(r) = 0 \;\wedge\; \alpha \Rf(r) \in [\tau_{\mathrm{lo}}, \tau_{\mathrm{hi}}]$ & \textsc{review} & rank by $\E[\text{loss}] = \sum_v \sev(v) h_v(\front_r)$ \\
$\Rb(r) > 0 \;\vee\; \alpha \Rf(r) > \tau_{\mathrm{hi}}$ & \textsc{pin} & keep on $V_1$; or synthesize migration \\
\bottomrule
\end{tabular}
\end{center}

Three policy consequences are worth stating.

\paragraph{The migrate action's mechanics deserve the same scrutiny as
its safety.} A \textsc{migrate} verdict presumes a working resume path:
the parked run must actually wake and re-enter execution under the new
deployment. Ground-truth testing of one stack found the wake endpoint
promoting deferred \emph{step}-queue messages while a current
\texttt{sleep} parks the run behind a deferred \emph{workflow}-queue
message, so the vendor-documented ``rerun on latest'' recovery silently
no-ops for sleeping runs. The model certifies that migration is safe;
whether the migration \emph{mechanism} fires is an orthogonal
operational property that should be verified per stack, ideally by the
same counterfactual testing the analyzer itself invites.

\paragraph{Pinning becomes a measured cost, not a default.} The expected
drain time of the pinned population is computable from the same chain
(expected absorption time from each frontier, dominated by residual timer
durations, which the logs record exactly). A release that strands runs
behind a 90-day sleep is now a quantified liability, informing the choice
between pinning and authoring an explicit state-migration function for the
$\Rb > 0$ residue.

\paragraph{Uncertainty routes to review.} Runs whose credible interval
straddles a threshold are classified by the interval's upper bound
(pessimistically) and surfaced with the diagnosis \emph{low-telemetry
branch}: the specific transitions whose posteriors are prior-dominated.
The remedy---observe more runs, or add a synthetic traversal in
staging---is concrete.

\paragraph{Release shaping.} By Proposition~\ref{prop:mono}, $\WUR$ is a
sub-release-additive objective: a planned change set can be greedily split
into a sequence of releases each below a risk budget, yielding an
upgrade plan rather than a single go/no-go.

The thresholds above assume independent runs. When the coupling graph of
\S\ref{sec:coupling} is nonempty, the partition is instead computed
jointly (Theorem~\ref{thm:cut}), and \textsc{review} is defined by label
instability under posterior draws (Remark~\ref{rem:review}).

\subsection{A Worked Example}\label{sec:example}

The whole per-run pipeline fits in one small, fully computed scenario
(numbers produced by the reference implementation). A linear order
workflow \texttt{validate} $\to$ \texttt{charge} $\to$
\texttt{sleep(24h)} $\to$ \texttt{notify} $\to$ \texttt{archive}
ships a $V_2$ in which \texttt{archive} newly consumes
\texttt{charge.promoCode}, a field present in $56.7\%$ of the 120
completed historical runs. Four in-flight runs, one change, three
different verdicts:

\begin{center}\small
\begin{tabular}{lcccl}
\hline
run (state at suspension) & $\Rb$ & $\Rf$ & $\WUR$ & verdict \\
\hline
$r_1$: before \texttt{charge} & $0.00$ & $0.43$ & $0.43$ &
  \textsc{review} \\
$r_2$: charged, \texttt{promoCode} recorded & $0.00$ & $0.00$ & $0.00$ &
  \textsc{migrate} \\
$r_3$: charged, \texttt{promoCode} absent & $0.00$ & $1.00$ & $1.00$ &
  \textsc{pin} \\
$r_4$: complete prefix (promo present) & $0.00$ & $0.00$ & $0.00$ &
  \textsc{migrate} \\
\hline
\end{tabular}
\end{center}

Every mechanism of the model is visible here. Backward risk is zero for
all four: the change adds no demand inside any recorded prefix, and
$r_4$'s completed \texttt{archive} validates against its recorded
payload. The forward term does the discriminating, and it discriminates
\emph{per run from recorded state}: $r_2$ and $r_3$ sit at the
identical frontier under the identical change, but $r_3$'s recorded
\texttt{charge} payload lacks the demanded field, so its future fault
is certain ($\Rf = 1.0$: hitting probability $1$ on a linear workflow
times an exact per-run incompatibility of $1$), while $r_2$'s is
impossible. Only $r_1$, whose \texttt{charge} has not yet executed, is
genuinely probabilistic---its $0.43$ is the telemetry-estimated field
absence rate, and its credible interval (not shown) is what routes it to
\textsc{review} rather than a coin flip. Note where the exact/probabilistic
boundary actually falls: not between backward and forward, but between
recorded and unrecorded state. Only $r_1$, whose demanded payload does
not yet exist, carries \emph{probabilistic} risk; $r_2$--$r_4$ are
decided exactly---$r_4$ by the backward validation of
Definition~\ref{def:rb} and $r_2$, $r_3$ by a forward term whose
incompatibility factor is an indicator over the invocation-time store
(\S\ref{sec:risk}), exact even though the fault it predicts lies in
the future. A rollout that pins $r_3$, migrates $r_2$ and $r_4$, and
holds $r_1$ for review strands exactly one run behind the old
deployment instead of the whole fleet.

\section{Inter-Run Coupling and the Mixed-Version Window}\label{sec:coupling}

This section is the operations layer of the roadmap in \S1: everything
above it produces per-run verdicts under an independence assumption, and
everything in it exists because production fleets violate that
assumption in two specific, telemetry-visible ways. Readers evaluating
only the per-run model can skip it; operators cannot, because a
partition that ignores coupling will split saga partners across versions
and manufacture exactly the mixed-version failures it was meant to
prevent.

One point of the timeline deserves stating plainly, because the phrase
``mixed-version window'' invites a misreading. The window is an
\emph{output} of the decision, not an input to it: the entire analysis
runs \emph{before any $V_2$ execution exists}, on $V_1$ telemetry and
the static diff alone, and the window opens as a consequence of its
verdict---migrated and fresh runs proceed on $V_2$ while the pinned
population drains on $V_1$. No canary phase or A/B observation period is
required to decide. What \emph{does} happen during the window is
re-decision, not first decision: $V_2$ telemetry accrues and sharpens
the posteriors (severity calibration across deployment history,
\S\ref{sec:severity}, consumes it automatically), the \textsc{review}
class resolves as intervals narrow---re-running the analyzer costs
microseconds per run---and pinned runs are re-evaluated when an
annotation, a rename map, or a state migration lands. The window closes
when the pinned population drains, on the schedule the pinning cost of
\S\ref{sec:policy} already prices.

The aggregate of Definition~\ref{def:wur} treats runs as independent.
Two mechanisms violate this. First, \emph{failure contagion}: runs
coupled through shared external effects can fail because a peer failed,
so individually certified runs are not jointly certified. Second---and
specific to this problem---\emph{mixed-version interaction}: any
nontrivial migrate/pin partition creates a window in which coupled runs
execute different code versions against shared channels. The partition
itself manufactures a version-skew surface at the data layer, so a
policy that ignores coupling can be individually optimal and jointly
unsafe. This section models both, reusing the machinery already in
place.

\subsection{The Coupling Graph}

\begin{definition}[Channel]\label{def:channel}
A \emph{channel} is a typed interaction surface between runs,
$c = (\mathrm{kind}, \mathrm{schema}, \mathrm{endpoints})$ with
$\mathrm{kind} \in \{\textsc{signal}, \textsc{hierarchy},
\textsc{resource}\}$: \textsc{signal} channels are hooks and events one
run raises and another awaits; \textsc{hierarchy} channels connect
parent and child runs; \textsc{resource} channels are pieces of external
state touched by steps of both endpoints, at least one as a writer.
\end{definition}

The \emph{coupling graph} $K = (\mathcal{R}, C)$ has runs as nodes and
channel instances as edges. Edges are estimated from three telemetry
tiers of decreasing certainty:

\begin{enumerate}[leftmargin=2em]
  \item \textbf{Protocol edges.} Hook dispatches, awaited events, and
  parent--child starts carry correlation tokens \emph{in the event logs}
  themselves; \textsc{signal} and \textsc{hierarchy} edges are therefore
  deterministic (confidence $1$), and their payloads are fully recorded.
  \item \textbf{Resource fingerprints.} Step instrumentation (trace-span
  attributes identifying tables, queues, endpoints) yields the resources
  each step touches and the access mode; two runs whose fingerprints
  intersect with at least one write induce a \textsc{resource} edge.
  \item \textbf{Key co-occurrence.} Entity identifiers shared across
  runs' recorded payloads (order, account, correlation keys) induce edge
  \emph{posteriors}: a Beta-distributed existence probability driven by
  co-occurrence statistics, entering the model with proportionally low
  confidence.
\end{enumerate}

\subsection{Failure Contagion as a Least Fixpoint}

\begin{definition}[Transmission]\label{def:trans}
For an edge $(r' \!\to\! r)$ over channel $c$, the \emph{transmission
probability} $\tau_{r'r} \in [0,1]$ is the probability that failure of
$r'$ leaves $c$ in a state that fails $r$. It factors as
$\tau_{r'r} = w(r',c)\cdot\kappa(c,r)$---did the failed run write the
channel, and is the channel critical to the reader---with per-kind Beta
priors updated by incident history where available.
\end{definition}

\begin{definition}[Contagion system]\label{def:contagion}
Given per-run endogenous risks $b_r$ (the coupled-free
$\WUR(r)$ of Equation~\eqref{eq:wur} under the chosen action for $r$),
the joint failure probabilities $q \in [0,1]^{\mathcal{R}}$ satisfy
\begin{equation}\label{eq:fix}
q_r \;=\; F_r(q) \;=\; 1 - (1 - b_r)
\prod_{r' \in N(r)} \bigl(1 - \tau_{r'r}\, q_{r'}\bigr).
\end{equation}
\end{definition}

\begin{proposition}[Least fixpoint]\label{prop:fix}
$F$ is monotone on the complete lattice $[0,1]^{\mathcal{R}}$, so a
least fixpoint $q^{*}$ exists~\cite{tarski} and equals
$\lim_k F^{k}(\bm{0})$; on acyclic coupling the iteration terminates in
graph-depth steps. Moreover $q^{*}$ is monotone in $(b, \tau)$, and
$\tau \equiv 0$ recovers the independent aggregate of
Definition~\ref{def:wur}.
\end{proposition}

Taking the \emph{least} fixpoint is a semantic commitment: every failure
must be grounded in a chain that bottoms out in endogenous risk.
Self-supporting failure cycles---$r$ fails because $r'$ does, and vice
versa, with no root cause---are excluded, in the spirit of well-founded
derivations. The fleet score becomes
$\WUR = \sum_r \omega_r\, q^{*}_r$, and by monotonicity coupling can
only increase it: independence was the optimistic bound.

\subsection{Mixed-Version Channel Risk}

A channel is an \emph{interface} whose writer and reader may sit on
different versions during the window. The machinery of
\S\ref{sec:changes} applies across the channel unchanged: with writer on
version $u$ and reader on version $w$,
\begin{equation}
\rho_{\times}(c;\, u\!\to\! w) \;=\;
\Prob_{x \sim \Phi_c}\!\bigl[\,x \nvDash D_w(c)\,\bigr],
\end{equation}
where $\Phi_c$ is the recorded payload distribution of the channel
(\textsc{signal} payloads are fully logged; \textsc{resource} shapes come
from fingerprint samples), and $D_w$ is the reader-side demand.
Lemma~\ref{lem:diffdemand} transfers verbatim: channels whose reader and
writer slices are both unchanged carry $\rho_{\times} = 0$ with no
inference at all.

Timestamps make the window quantitative. The channel's event rate
$\lambda_c$ carries a Gamma posterior from inter-event times in
$\mathcal{H}$; the expected overlap duration $T(a)$ under partition $a$
is the drain time of the pinned side (absorption times plus residual
timer durations, which the logs record exactly). Under Poisson arrivals,
the \emph{mixed-window penalty}
\begin{equation}\label{eq:penalty}
\pi_c(a) \;=\; 1 - \exp\bigl(-\lambda_c\, T(a)\, \rho_{\times}(c)\bigr)
\end{equation}
is the probability of at least one incompatible cross-version
interaction on $c$ during the window.

\subsection{The Coupling-Aware Partition Is a Minimum Cut}

The migrate/pin decision is now a joint labeling
$a \in \{V_1, V_2\}^{\mathcal{R}}$ minimizing
\begin{equation}\label{eq:energy}
E(a) \;=\; \sum_{r} \psi_r(a_r) \;+\;
\sum_{c = (r, r') \in C} \pi_c \,\ind{a_r \neq a_{r'}},
\end{equation}
where $\psi_r(V_2)$ is the run's migration risk (with
$\Rb(r) > 0 \Rightarrow \psi_r(V_2) = \infty$: hard pin) and
$\psi_r(V_1)$ is the \emph{pinning cost}: stranding liability (expected
drain time) plus forgone-fix exposure $\beta_r$---the probability that
the run reaches the defect $V_2$ repairs, which is itself a hitting
probability $h_{\mathrm{defect}}(\front_r)$ from the same linear solve
as \S\ref{sec:risk}. The reason to upgrade and the risk of upgrading are
computed by the same operator.

\begin{theorem}[Exact partition]\label{thm:cut}
With $\pi_c \geq 0$ penalizing only label disagreement, the pairwise
terms of \eqref{eq:energy} are submodular for binary labels, and the
global minimum of $E$ is computed exactly in polynomial time by a single
$s$--$t$ minimum cut~\cite{greig,kz}. Must-move-together constraints
($\pi_c = \infty$, e.g.\ saga partners over deterministic channels) are
handled by contracting their endpoints into super-nodes beforehand.
\end{theorem}

Figure~\ref{fig:cut} shows the construction on four runs: a coupled pair
whose channel penalty exceeds either run's assignment cost stays on one
side of the cut together, which is exactly the mixed-version-window
discipline the penalty encodes.

\begin{figure}[t]
\centering
\begin{tikzpicture}[node distance=7mm and 15mm]
\node[draw, circle, font=\scriptsize, inner sep=1.6pt] (s) {$s$};
\node[lab, above=4mm of s, align=center] {\textsc{migrate}\\ ($V_2$)};
\node[draw, circle, font=\scriptsize, right=58mm of s, inner sep=1.6pt]
  (t) {$t$};
\node[lab, above=4mm of t, align=center] {\textsc{pin}\\ ($V_1$)};
\node[site, right=of s, yshift=15.5mm, minimum width=8mm] (r1) {$r_1$};
\node[site, below=of r1, minimum width=8mm] (r2) {$r_2$};
\node[site, below=of r2, minimum width=8mm] (r3) {$r_3$};
\node[site, below=of r3, minimum width=8mm] (r4) {$r_4$};
\foreach \r in {1,2,3,4}{
  \draw[flow] (s) -- (r\r.west)
    node[tinylab, pos=0.82, above, sloped] {$\psi_{r_\r}(V_1)$};
  \draw[flow] (r\r.east) -- (t)
    node[tinylab, pos=0.18, above, sloped] {$\psi_{r_\r}(V_2)$};}
\draw[thick] (r1.south) -- (r2.north)
  node[tinylab, midway, left=1pt] {$\pi_{12}$};
\draw[thick] (r3.south) -- (r4.north)
  node[tinylab, midway, left=1pt] {$\pi_{34}$};
\draw[very thick, densely dashed]
  ($(r1.north east)+(14mm,3mm)$)
  .. controls ($(r2.east)+(12mm,-1mm)$) and ($(r3.west)+(-10mm,1mm)$) ..
  ($(r4.south west)+(-14mm,-3mm)$);
\node[lab, anchor=north, align=center] at
  ($(s.south)!0.5!(t.south)+(0,-27mm)$)
  {cut: $\{r_1,r_2\}\!\to\!V_2$ severing their edges to $t$;\ \
   $\{r_3,r_4\}\!\to\!V_1$ severing their edges from $s$};
\end{tikzpicture}
\caption{The migrate/pin partition as a single $s$--$t$ minimum cut.
Cutting a terminal edge pays the run's assignment cost; separating a
coupled pair (here saga partners $r_1, r_2$ and $r_3, r_4$ over
deterministic channels) pays the channel penalty $\pi_c$, so tightly
coupled runs cross versions only when both assignment costs justify it.
Runs with $\Rb > 0$ carry $\psi(V_2) = \infty$ and can never land on the
migrate side.}
\label{fig:cut}
\end{figure}

\begin{proof}[Proof sketch]
Standard construction: terminal edges carry the unary costs
$\psi_r(\cdot)$, neighbor edges carry $\pi_c$. Submodularity
$\theta(0,1) + \theta(1,0) \geq \theta(0,0) + \theta(1,1)$ holds since
disagreement costs $\pi_c \geq 0$ and agreement costs $0$; exactness of
the cut for submodular binary energies is~\cite{kz}.
\end{proof}

\begin{remark}[Review as label instability]\label{rem:review}
Under coupling, the \textsc{review} class is defined jointly rather than
by per-run thresholds: draw $(\rho, \tau, \lambda, P)$ from their
posteriors, re-solve the cut per draw, and route to \textsc{review}
exactly the runs whose optimal label flips across draws. The stable core
of the cut migrates or pins with quantified confidence.
\end{remark}

\begin{remark}[More than two versions]
With $k > 2$ versions live, \eqref{eq:energy} becomes a multiway-cut
energy, NP-hard in general; $\alpha$-expansion~\cite{bvz} gives strong
local optima with known approximation bounds. The operational reading is
a design rule the model makes precise: bound the number of concurrently
live versions to two, and exactness is preserved.
\end{remark}

Computationally the extension is light: the fixpoint iteration costs
$O(\mathrm{iters} \cdot |C|)$ with fast convergence in practice
(coupling graphs are sparse and shallow), and the cut is one max-flow on
$|\mathcal{R}|$ nodes---both dominated by the payload validation already
in Proposition~\ref{prop:complexity}.

\section{Limitations and Extensions}\label{sec:limitations}

\emph{Behavioral semantics.} $\Rf$ measures exposure to changed code, not
whether the change is desirable---an intentional bug fix and an
accidental regression score identically. Severity calibration
(\S\ref{sec:severity}) weights exposure by empirical harm, but cannot
distinguish intent.

\emph{Unobserved paths.} Branches never traversed in $\mathcal{H}$ carry
prior-dominated posteriors; new sites in $V_2$ have no telemetry at all.
The model is honest about this (wide intervals) but not clairvoyant.

\emph{The Markov assumption.} The forward chain is first-order over
sites; workflows whose routing depends on accumulated state (retry
budgets, aggregated results) have history-dependent transitions the
chain averages over. The hitting probabilities remain unbiased in
aggregate but can misrank individual runs; conditioning the chain on
frontier payload features (a semi-Markov or feature-augmented chain) is
a mechanical extension the estimation machinery of
\S\ref{sec:estimation} already supports, at the cost of thinner
per-state telemetry.

\emph{Cold start and priors.} Until history accumulates, severity and
transition posteriors are prior-dominated; the honest behavior is wide
intervals routing to \textsc{review}, not confident verdicts from
invented constants. Teams that must act before telemetry exists should
set deliberately conservative priors per change class and let the
credible-interval width report the cost of ignorance transparently; the
cold-start measurements of \S\ref{sec:evaluation} show the
distributional regime becoming informative within tens of completed
runs.

\emph{Cyclic coupling.} On dense coupling cycles the contagion fixpoint
overestimates (\S\ref{sec:evaluation}), always in the safe direction;
tighter cyclic estimates via mean-field corrections or loopy belief
propagation with convergence guarantees are future work.

\emph{Production validation.} The evaluation demonstrates soundness
against verified semantics and relevance on the reference repository's
own workflow corpus; it does not yet include production fleets, because
the telemetry the model consumes---run listings, event logs, recorded
payloads at fleet scale---is exactly the surface only a platform
operator holds. This is a structural property of the approach, not an
accident of the study: the model is designed to run \emph{where the
protocol data already lives}. Validating it against a production
deployment history (releases with known outcomes, replayed through the
analyzer counterfactually) requires operator collaboration and is the
single highest-value next step.

\emph{Order-observed concurrency.} \texttt{Promise.race}/\texttt{any} and
incremental awaiting retain order semantics by construction
(Definition~\ref{def:indep}); permutation sensitivity there is real, not
spurious, so such regions are simply less migrate-eligible. This is a
correct restriction, not an unsoundness, but it does reduce the
\textsc{migrate} class in racy workflows.

\emph{Dynamic demand.} Computed property accesses and reflective
consumption widen verdicts to \textsc{unk}
(\S\ref{sec:demand}), exerting \textsc{review} pressure in heavily
dynamic code. Lemma~\ref{lem:diffdemand} confines this to changed slices,
but a release that rewrites a reflective consumer will legitimately
resist certification.

\emph{Cold start.} Severity class calibration requires deployment
history; a fresh installation inherits cross-fleet priors and wide
intervals until its own releases accumulate.

\emph{Unobserved channels.} Coupling through systems that neither the
protocol nor step instrumentation observes is invisible; the coupling
graph is a lower bound on true coupling, and \eqref{eq:fix} correspondingly
a lower bound on joint risk. Instrumentation coverage is thus a direct
input to certification strength.

\emph{Transmission estimates.} $\tau$ is the weakest posterior in the
model: incidents are rare, so priors dominate for long. Conservative
per-kind defaults inflate \textsc{review}, never \textsc{migrate}.

\emph{Many concurrent versions.} The exact partition is binary;
$k > 2$ live versions require approximate multiway optimization
(\S\ref{sec:coupling}).

\section{Evaluation}\label{sec:evaluation}

We evaluate a reference implementation of the model against a synthetic
suite built around a ground-truth oracle that reproduces the verified
replay semantics of Remark~\ref{rem:conformance}: one identifier sequence
shared by all site kinds, ordinal-plus-name matching, in-order log
consumption, invocation-time demand evaluation, hard divergence on
unmatched events. A single scheduler drives both run simulation and
oracle replay, so replaying any log against its own version is safe by
construction (verified on 1{,}500 runs). The analyzer sees only the
static diff and telemetry; it never executes $V'$ against a log. All code
and experiment configurations accompany the paper.

\subsection{Scenario suite}

Eleven controlled scenarios cover the change classes of
\S\ref{sec:changes} over a fixed workflow with a parallel group, a
branch, timers, and partial-presence payload fields; each scenario is
evaluated against 400 in-flight runs suspended at uniformly random
frontiers, with the ground truth given by oracle replay of every run:

\begin{center}\small
\begin{tabular}{lrrrr}
\hline
scenario & predicted & true & FP & FN \\
\hline
S1 body-only edit & 0 & 0 & 0 & 0 \\
S2 consume always-present field & 0 & 0 & 0 & 0 \\
S3 consume partial field ($p{=}0.7$) & 103 & 103 & 0 & 0 \\
S4 insert step at head & 348 & 348 & 0 & 0 \\
S5 insert step at tail & 0 & 0 & 0 & 0 \\
S6 delete tail step & 36 & 36 & 0 & 0 \\
S7 swap step/timer & 307 & 307 & 0 & 0 \\
S8 \texttt{Promise.all} branch reorder & 194 & 194 & 0 & 0 \\
S9 predicate semantic change & 86 & 33 & 53 & 0 \\
S10 predicate refactor (no-op) & 86 & 0 & 86 & 0 \\
S11 insert timer mid-workflow & 194 & 194 & 0 & 0 \\
\hline
\end{tabular}
\end{center}

The exact-regime verdict agrees with the oracle perfectly on all
structural scenarios, including the case
Remark~\ref{rem:conformance} makes identity-sensitive: reordering the
branches of a join-observed parallel group (S8) is runtime-fatal for
every run whose prefix reaches the group, and the analyzer classifies
each of the 194 identically. S4 versus S5 shows the prefix rule at its
tightest: an insertion at the head invalidates every non-empty prefix,
while the same insertion at the tail invalidates none. The two predicate
scenarios exhibit the designed conservatism of syntactic matching: a
semantic threshold change (S9) is flagged for all 86 runs that crossed
the branch while 33 truly diverge, and a pure refactor (S10, identical
semantics, different syntax) is flagged for the same 86 while none
diverge. No scenario produces a false negative.

\subsection{Mutation study}

We generated 120 random workflows (nested parallel groups, branches,
timers, hooks, partial-presence fields) and applied every applicable
mutation operator from a pool of ten covering all change classes,
including producer deletions, producer--consumer order inversions, and
reorders of both join-observed and race branches, yielding 1{,}082
mutations evaluated against 60 in-flight runs each (64{,}920 run-level
verdicts):

\begin{center}\small
\begin{tabular}{lrrrr}
\hline
mutation class & TP & FP & FN & TN \\
\hline
body-only edit & 0 & 0 & 0 & 7{,}200 \\
demand growth (present field) & 519 & 296 & 0 & 6{,}385 \\
demand growth (partial field) & 1{,}061 & 268 & 0 & 5{,}751 \\
insert step & 3{,}686 & 0 & 0 & 3{,}514 \\
insert timer & 3{,}296 & 0 & 0 & 3{,}904 \\
delete site & 3{,}553 & 24 & 0 & 3{,}623 \\
swap adjacent sites & 4{,}321 & 15 & 0 & 2{,}864 \\
parallel/race branch reorder & 2{,}276 & 31 & 0 & 2{,}373 \\
predicate change (semantic) & 655 & 1{,}535 & 0 & 2{,}790 \\
predicate refactor (no-op) & 0 & 2{,}240 & 0 & 2{,}740 \\
\hline
total & 19{,}367 & 4{,}409 & 0 & 41{,}144 \\
\hline
\end{tabular}
\end{center}

Recall is exactly $1.0$---zero false negatives across all 64{,}920
verdicts. The epistemic status of this number deserves precision: it
validates that the implementation satisfies Theorem~\ref{thm:sound}
against an oracle enforcing the verified replay semantics; it is a
soundness check of the model against its own semantics, not a claim that
the model catches failures outside its scope (violated determinism
contracts, inference front-end errors on dynamic access, or
non-protocol side effects; see \S\ref{sec:limitations}). One class is
outside the oracle's observable surface by construction: the silent
result-crossing of Definition~\ref{def:matching}'s same-label swap
replays as \textsc{safe}---that is what \emph{silent} means---so
evaluating its detection requires a semantic oracle comparing final
states against fresh $V_2$ execution, which we leave as future work; the
fingerprint requirement that flags it statically is stated but not
exercised by this suite. Within scope,
perfect recall is not surprising---it is what the theorem requires, and
a single false negative would be an implementation bug, as three found
during development were.
Precision is $0.815$, and every false positive has an attributable
cause:
syntactic predicate matching (86\%), demands whose satisfiability crosses
conditional-arm or race-branch boundaries the scope analysis cannot prove
(13\%), and the conservative certainly-precedes prefix rule
(2\%)---i.e., the deliberate conservatism analyzed in
\S\ref{sec:limitations}, not the model. Precision has a direct
operational reading: $1 - \mathrm{precision}$ is the fraction of
fatal verdicts that needlessly pin a run, and its cost is exactly the
stranding liability and forgone-fix exposure that
\S\ref{sec:policy} prices---so the conservatism is not free, but it is
\emph{measured}, and every attributable cause above is also an
annotation target (predicate-equivalence hints alone would remove
86\% of it). Notably, demand-growth mutations include consumers that reach
across parallel branches; catching these requires the invocation-time
store of Definition~\ref{def:rb}, and validating against the
end-of-prefix store instead produces false negatives (an error mode this
suite caught in an earlier revision of the model).

\subsection{Real-workflow corpus}

Synthetic topologies establish soundness; real ones establish relevance.
The corpus is the workflow code shipped in the reference implementation's
own repository, \texttt{github.com/\allowbreak vercel/\allowbreak workflow} at commit
\href{https://github.com/vercel/workflow/tree/39673b7}{\texttt{39673b7}}
(2026-07-08): every function carrying the
\texttt{'use workflow'} directive under
\texttt{workbench/\allowbreak vitest/\allowbreak workflows} (including the \texttt{cookbook}
suite of documented canonical patterns) and under
\texttt{workbench/\allowbreak example/\allowbreak workflows} files \texttt{0}--\texttt{10},
excluding only the load/e2e stress files (\texttt{97}--\texttt{100})
whose hundred-odd micro-workflows exercise the runtime rather than
represent applications. This yields 45 workflow functions; the model
corpus has 44 entries (one thin child-spawning wrapper is folded into its
parent as the step invocation it is), of which 41 contain at least one
durable site, totaling 126 sites.

{\sloppy ``Hand-translated'' means the following mechanical procedure, whose
output is a reviewable artifact (one annotated constructor per workflow,
keyed by source path, in \texttt{real\_\allowbreak corpus.py} accompanying the
paper). (1)~Each \texttt{'use step'} call maps to a step whose payload
schema is the step body's \emph{actual} return-object fields;
TypeScript-optional or conditionally assigned fields get partial presence
probabilities, all others presence~1. (2)~\texttt{sleep()} maps to a
timer; \texttt{createWebhook()} or \texttt{hook.create()} followed by
\texttt{await} maps to a payload-bearing hook. (3)~\texttt{Promise.all}
and \texttt{allSettled} map to join-observed parallel groups with
branches in source-array order (order is identity);
\texttt{Promise.race} maps to a race group with the verified semantics
of Remark~\ref{rem:conformance}---branches invoked in fork order, the
continuation resumed at the first completion in log order, losers left
pending, consumable if their completions appear later in the log and
dropped at workflow return otherwise. A background invocation awaited
after intervening work is an additional branch at its invocation point.
(4)~\texttt{if}/\texttt{else} on step or hook results maps to a branch
whose predicate reads the recorded payload field the source reads;
\texttt{try}/\texttt{catch} compensation (the saga pattern) maps to a
branch into the compensation sequence in unwind order. (5)~Data-flow
demand edges are placed exactly where the source passes one site's result
into another call's arguments. (6)~Loops with source-fixed bounds are
unrolled at those bounds. Two mechanical checks validate the
translations: every one of 1{,}320 prefixes of simulated runs replays
\textsc{safe} against its own topology under the oracle, and the race
workflows exhibit all four log shapes the runtime can record (winner
only, winner then loser, loser interleaved after post-race sites, either
branch winning). Consumers reaching into race branches are treated by the
scope analysis as never-guaranteed, since a loser may never complete.
The known approximations are stated rather than hidden: loop unrolling
fixes iteration counts, and winner-dependent orchestrator transforms
(\texttt{.then} mappings on race results) are deterministic orchestrator
code, not payload demands, so they carry no demand edge.\par} The corpus spans the patterns the vendor
documents as canonical: sagas, fan-out, human-in-the-loop approval,
parent--child workflows over completion hooks, content routing, batching
with inter-batch timers, webhooks, and durable agents.

Applying the mutation pool to this corpus yields 227 mutations and
18{,}160 run-level verdicts (80 in-flight runs per workflow): recall is
again exactly $1.0$, precision $0.933$.

\begin{center}\small
\begin{tabular}{lrrrr}
\hline
mutation class (real corpus) & TP & FP & FN & TN \\
\hline
body-only edit & 0 & 0 & 0 & 3{,}120 \\
demand growth (present field) & 330 & 97 & 0 & 1{,}093 \\
demand growth (partial field) & 30 & 0 & 0 & 290 \\
insert step & 1{,}218 & 0 & 0 & 2{,}062 \\
insert timer & 1{,}142 & 0 & 0 & 2{,}138 \\
delete site & 1{,}362 & 0 & 0 & 1{,}518 \\
swap adjacent sites & 1{,}319 & 0 & 0 & 441 \\
parallel/race branch reorder & 328 & 10 & 0 & 222 \\
predicate change (semantic) & 12 & 141 & 0 & 567 \\
predicate refactor (no-op) & 0 & 162 & 0 & 558 \\
\hline
total & 5{,}741 & 410 & 0 & 12{,}009 \\
\hline
\end{tabular}
\end{center} Precision is higher than on the
synthetic suite because real topologies are shorter and their prefixes
more often provably unaffected. The corpus also improved the model:
producer deletion and producer--consumer order inversion---natural here,
where hook payloads feed later steps---produced false negatives in an
earlier revision, forcing the demand-satisfiability rule and the
$V_2$-side invocation gate of \S\ref{sec:demand}. The residual false
positives are again fully attributable: syntactic predicate matching
(74\%, concentrated in the router and single-statement control-flow
workflows), demand scope across conditional arms and race branches
(24\%), and the certainly-precedes rule on branch reorders (2\%).

\subsection{Calibration of distributional interface risk}

For 400 trials with the demanded field's presence probability drawn
uniformly from $[0.05, 0.95]$, the distributional-regime $\Rb$ estimated
from historical presence telemetry was compared with the empirical
replay-failure rate among affected in-flight runs. Expected calibration
error over ten bins is $0.016$, Pearson correlation $0.92$, mean absolute
error $0.087$---consistent with the claim that when payloads are not
retained, presence frequencies alone give a well-calibrated risk.

\subsection{Contagion fixpoint versus simulated cascades}

We compared the least-fixpoint estimate \eqref{eq:fix} against
Monte-Carlo cascade simulation (20{,}000 samples per graph) on random
coupling graphs, five graphs per configuration, reporting the signed
error $q^* - \hat q_{\mathrm{MC}}$ per node:

\begin{center}\small
\begin{tabular}{llrrr}
\hline
topology & $(n, \bar d)$ & worst over & worst under & mean $|{\cdot}|$ \\
\hline
acyclic & $(30, 1.5)$ & $0.053$ & $-0.009$ & $0.005$ \\
acyclic & $(30, 3.0)$ & $0.095$ & $-0.006$ & $0.016$ \\
acyclic & $(100, 2.0)$ & $0.072$ & $-0.009$ & $0.004$ \\
acyclic & $(100, 4.0)$ & $0.084$ & $-0.008$ & $0.014$ \\
cyclic & $(30, 1.5)$ & $0.325$ & $-0.002$ & $0.160$ \\
cyclic & $(30, 3.0)$ & $0.091$ & $-0.002$ & $0.035$ \\
cyclic & $(100, 2.0)$ & $0.273$ & $-0.004$ & $0.050$ \\
cyclic & $(100, 4.0)$ & $0.025$ & $-0.005$ & $0.002$ \\
\hline
\end{tabular}
\end{center}

The ``worst under'' column never drops below $-0.009$, which is within
Monte-Carlo noise at 20{,}000 samples: empirically the fixpoint is a
\emph{one-sided, conservative} estimate, consistent with the positive
association of percolation-style cascade events (Harris/FKG), under which
the noisy-or independence assumption can only overestimate. On acyclic
coupling the bound is tight; on cyclic coupling, feedback echo inflates
the estimate, worst where a sparse cycle concentrates the echo.
Overestimation moves runs from \textsc{migrate} toward
\textsc{review}/\textsc{pin}, never the reverse, so the approximation
errs in the safe direction; tighter cyclic estimates are future work.

\subsection{Partition optimality and cost}

Across 200 random instances of 12 runs with dense coupling, the min-cut
partition of \S\ref{sec:coupling} attains exactly the brute-force
minimum energy (maximum gap $0.0$ over all $2^{12}$ labelings per
instance), as Theorem~\ref{thm:cut} requires. One implementation
note: computing the partition side of the cut with floating-point
capacities corrupted residual reachability in a standard max-flow
library; integer-scaled capacities restore exactness.

Backward risk in the exact regime costs 9--15\,$\mu$s per run in a pure
Python reference implementation (single core): a fleet of 50{,}000
in-flight runs is classified in $0.72$\,s, three orders of magnitude
inside the per-release budget of Proposition~\ref{prop:complexity},
before any of the prefix-hash deduplication that production fleets
admit. Scaling in workflow \emph{size} is linear in prefix length, as
Proposition~\ref{prop:complexity} predicts: per-run $\Rb$ grows from
$4\,\mu$s at 10 sites to $3.6$\,ms at a pathological 2{,}000 sites
(with the one-time static diff at $0.5$\,s), so even a fleet of
10{,}000 runs of thousand-step workflows classifies in under a minute
on one core. Telemetry volume matters least where one might fear it
most: the distributional regime's error against ground truth falls from
MAE $0.15$ with only ten completed historical runs to $0.09$ at three
hundred, at which point the residual is the binomial noise of the
affected-run sample rather than the estimator---useful predictions
begin at tens of runs, not thousands.

\section{Conclusion}

Durable workflow engines already pay the storage cost of complete
execution histories; this paper converts that cost into an analytical
asset. Because the event log is the protocol, upgrade risk in this domain
splits cleanly into a component that can be computed exactly from recorded
prefixes and payloads, and a component that is genuinely probabilistic and
estimable, with honest uncertainty, from complete path histories. The
resulting $\WUR$ score turns version-skew handling from a uniformly
pessimistic mechanism into a measured decision: migrate what is provably
safe, review what is uncertain, pin what is provably not---and know, in
advance, what each release will cost the fleet.

\paragraph{Availability.}
The reference implementation (a single-file analyzer consuming only
\texttt{runs.list} and \texttt{events.list}), the hand-translated
real-workflow corpus keyed to source paths at the stated commit, the
full experiment suite, and the raw results accompany this paper. We
invite platform operators holding production deployment histories to
collaborate on counterfactual validation---replaying past releases with
known outcomes through the analyzer---which is the highest-value
experiment this paper cannot run alone; contact the authors.

\end{document}